\documentclass[final]{IEEEtran}
\usepackage{amsmath,epsfig,times}
\usepackage{amsthm,amsfonts,amssymb}
\usepackage{booktabs}
\usepackage{graphicx}
\usepackage{rotating}
\usepackage{floatflt} 
\usepackage{paralist} 
\usepackage[dvipsnames]{xcolor}
\usepackage{color}
\usepackage{epstopdf}
\usepackage[normalem]{ulem}
\usepackage{url}
\usepackage{todonotes}
\usepackage{comment}
\usepackage{bm, bbm}
\usepackage{subcaption}
\usepackage{algorithm}
\usepackage{algpseudocode}
\usepackage{enumitem}
\usepackage{multirow, hhline, tabularx}
\usepackage{hyperref}
\usepackage{setspace}

\def\x{{\mathbf x}}

\def\Dk{{\bm{\Delta}^{(k+1)}}}

\def\y{{\bm{y}}}
\def\x{{\bm{x}}}
\def\betab{{\bm{\beta}}}

\def\z{{\bm{z}}}
\def\Z{{\bm{Z}}}
\def\U{{\bm{U}}}
\def\R{{\mathbb{R}}}
\def\D{{\bm{D}}}
\def\B{{\bm{B}}}
\def\betastar{{\betab^{\star}}}
\def\betahat{{\widehat{\betab}}}
\def\Bhat{\bm{\widehat{B}}}
\def\Bstar{\B^\star}

\newcommand{\rlam}[1]{\rho({#1})}
\newcommand{\glam}[1]{g({#1})}
\newcommand{\hlam}[1]{h({#1})}

\DeclareMathOperator{\argmin}{argmin}

\def\Ddag{{\bm{D}^{\dagger}}}
\def\Dkdag{{\bm{\Delta}^{(k+1)\dagger}}}

\newcommand{\Adj}{\bm{A}}

\def\V{\mathcal{V}}
\def\E{\mathcal{E}}

\newtheorem{defn}{Definition}
\newtheorem{remark}{Remark}
\newtheorem{assumptions}{Assumption}

\newtheorem{myThm}{Theorem}

\newtheorem{myProposition}{Proposition}

\algnewcommand{\Inputs}[1]{%
  \State \textbf{Inputs:}
  \Statex \hspace*{\algorithmicindent}\parbox[t]{.8\linewidth}{\raggedright #1}
}

\algnewcommand{\Initialize}[1]{%
  \State \textbf{Initialize:}
  \Statex \hspace*{\algorithmicindent}\parbox[t]{.8\linewidth}{\raggedright #1}
}

\title{Vector-Valued Graph Trend Filtering with Non-Convex Penalties}

\author{Rohan Varma, \IEEEmembership{Student Member}, Harlin Lee, \IEEEmembership{Student Member},  Jelena~Kova\v{c}evi\'c, \IEEEmembership{Fellow}, and\\ Yuejie Chi, \IEEEmembership{Senior Member}
\thanks{R. Varma, H. Lee and Y. Chi are with the Dept. of Electrical and Computer Engineering, Carnegie Mellon University. Emails:\{rohanv, harlinl, yuejiec\}@andrew.cmu.edu.}
\thanks{J. Kova\v{c}evi\'c is with the Tandon School of Engineering, New York University. Email: jelenak@nyu.edu. }%
\thanks{This work is supported in part by NSF under grants CCF-1563918, CCF-1826519, CCF-1806154 and ECCS-1818571, by ONR under grant N00014-18-1-2142, by ARO
under grant W911NF-18-1-0303, and by NIH under grant R01EB025018. A preliminary version of partial results in this paper was presented at the 2019 IEEE International Conference on Acoustics, Speech and Signal Processing \cite{varma2019improving}.}%
\thanks{The first two authors contributed equally.}}

\begin{document}

\maketitle

\begin{abstract}
    This work studies the denoising of piecewise smooth graph signals that exhibit inhomogeneous levels of smoothness over a graph, where the value at each node can be vector-valued. We extend the graph trend filtering framework to denoising vector-valued graph signals with a family of non-convex regularizers, which exhibit superior recovery performance over existing convex regularizers. Using an oracle inequality, we establish the statistical error rates of first-order stationary points of the proposed non-convex method for generic graphs. Furthermore, we present an ADMM-based algorithm to solve the proposed method and establish its convergence. Numerical experiments are conducted on both synthetic and real-world data for denoising, support recovery, event detection, and semi-supervised classification.
\end{abstract}


\section{Introduction}

Signal estimation from noisy observations is a classic problem in signal processing and has applications in signal inpainting, collaborative filtering, recommendation systems and other large-scale data completion problems. Since noise can have deleterious, cascading effects in many downstream tasks, being able to efficiently and accurately filter and reconstruct a signal is of significant importance. 

With the explosive growth of information and communication, signals are generated at an unprecedented rate from various sources, including social networks, citation networks, biological networks, and physical infrastructure~\cite{newman_networks_2018}. Unlike time series or images, these signals lie on complex, irregular structures, and require novel processing techniques, leading to the emerging field of {\em graph signal processing} \cite{sandryhaila_discrete_2013}\nocite{shuman_emerging_2013}-\cite{ortega_graph_2018}. This framework models the structure by a graph and generalizes concepts and tools from classical discrete signal processing to graph signal processing. The associated graph-structured data are referred to as \emph{graph signals}. 

In graph signal processing, a common assumption is that the graph signal is smooth with respect to the graph, that is, the signal coefficients do not vary much over local neighborhoods of the graph. However, this characterization is insufficient for many real-world signals that exhibit spatially inhomogeneous levels of smoothness over the graph. In social networks for example, within a given community or social circle, users' profiles tend to be homogeneous, while within a different social circle they will be of different, yet still have homogeneous values. Consequently, the signal is often characterized by large variations between regions and small variations within regions such that there are localized discontinuities and patterns in the signal. As a result, it is necessary to develop representations and algorithms to process and analyze such {\em piecewise smooth} graph signals. 
	
In this work, we study the denoising of the class of piecewise smooth graph signals (including but not limited to piecewise constant graph signals), which is complementary to the class of smooth graph signals that exhibit homogeneous levels of smoothness over the graph. The reconstruction of smooth graph signals has been well studied in previous work both within graph signal processing \cite{shuman_emerging_2013}\nocite{chen_signal_2015,ortega_graph_2018,chen2015discrete,chen2015signal,romero2016kernel}-\cite{elmoataz_nonlocal_2008} as well as in the context of Laplacian regularization \cite{belkin_regularization_2004,zhu_semi-supervised_2003}. 

The Graph Trend Filtering (GTF) framework~\cite{wang_trend_2016}, which applies total variation denoising to graph signals \cite{kim_ell_1_2009}, is a particularly flexible and attractive approach that regularizes discrete graph differences using the $\ell_1$ norm. Although the $\ell_1$ norm based regularization has many attractive properties~\cite{buhlmann_statistics_2011}, the resulting estimates are biased toward zero for large coefficients. To alleviate this bias effect, non-convex penalties such as the Smoothly Clipped Absolute Deviation (SCAD) penalty~\cite{fan_variable_2001} and the Minimax Concave Penalty (MCP)~\cite{zhang_nearly_2010} have been proposed as alternatives.
These penalties behave similarly to the $\ell_1$ norm when the signal coefficients are small, but tend to a constant when the signal coefficients are large. Notably, they possess the so-called \emph{oracle property}: in the asymptotics of large dimension, they perform as well as the case where we know in advance the support of the sparse vectors~\cite{loh_regularized_2013}\nocite{loh_statistical_2017, zhang_general_2012,breheny_coordinate_2011}-\cite{ma_concave_2017}. 

In this work, we strengthen the GTF framework in~\cite{wang_trend_2016} by considering a large family of possibly non-convex regularizers, including SCAD and MCP that exhibit superior reconstruction performance over $\ell_1$ minimization for the denoising of piecewise smooth graph signals. Furthermore, we extend the GTF framework to allow vector-valued signals, e.g. time series~\cite{chen2014semi}, on each node of the graph, which greatly broadens the applicability of GTF to applications in social networks \cite{hallac_network_2015}, gene networks, and semi-supervised classification \cite{ jung2016semi}\nocite{jung2018network}-\cite{tran2017network}.  
Through theoretical analyses and empirical performance, we demonstrate that the use of non-convex penalties improves the performance of GTF in terms of both reduced reconstruction error and improved support recovery, i.e. how accurately we can localize the discontinuities of the piecewise smooth signals. Our contributions can be summarized as follows:

\begin{itemize}
\item Theoretically, we derive the statistical error rates of the signal estimates, defined as first-order stationary points of the proposed GTF estimator. We derive the rates in terms of the noise level and the alignment of the ground truth signal with respect to the underlying graph, without making assumptions on the piecewise smoothness of the ground truth signal. The better the alignment, the more accurate the estimates. Importantly, the estimators do not need to be the global minima of the proposed non-convex problem, which are much milder requirements and important for the success of optimization. For denoising vector-valued signals, the GTF estimate is more accurate when each dimension of the signal shares similar patterns across the graph.
\item Algorithmically, we propose an ADMM-based algorithm that is guaranteed to converge to a critical point of the proposed GTF estimator.
\item Empirically, we demonstrate the performance improvements of the proposed GTF estimators with non-convex penalties on both synthetic and real data for signal estimation, support recovery, event detection, and semi-supervised classification.
\end{itemize}

The rest of this paper is organized as follows. Section~\ref{sec:related_works} reviews related works and their relationships to the current paper. In Section~\ref{sec:basics}, we provide some background and definitions on graph signal processing and GTF. Section~\ref{sec:gtf_nonconvex} presents the proposed GTF framework with non-convex penalties and vector-valued graph signals. Section~\ref{sec:theory} develops its performance guarantees, and Section~\ref{sec:admm} presents an efficient algorithm based on ADMM. Numerical performance of the proposed approach is examined on both synthetic and real-world data for denoising and semi-supervised classification in Section~\ref{sec:numerical}. Finally, we conclude in Section~\ref{sec:conclusions} and briefly discuss future work.

Throughout this paper, we use boldface letters $\bm{a}$ and $\bm{A}$ to represent vectors and matrices respectively. The transpose of $\bm{A}$ is denoted as $\bm{A}^{\top}$. The $\ell$-th row of a matrix $\bm{A}$ is denoted as $\bm{A}_{\ell \cdot}$, and the $j$-th column of a matrix $\bm{A}$ is denoted as $\bm{A}_{ \cdot j}$. The cardinality of a set $T$ is denoted as $|T|$. For any set $T \subseteq \{1, 2, ..., r\}$ and $\x \in \R^r$, we denote $(\x)_T \in \R^{|T|}$ such that $x_\ell \in (\x)_T$ if and only if $\ell \in T$ for $\ell \in \{1, 2, ..., r\}$. Similarly, we define a submatrix $\bm{A}_{T\cdot}\in\mathbb{R}^{|T|\times d}$ of $\bm{A}\in\mathbb{R}^{r \times d}$ that corresponds to pulling out the rows of $\bm{A}$ indexed by $T$. The $\ell_2$ norm of a vector $\bm{a}$ is defined as $\|\bm{a}\|_2$, and the spectral norm of a matrix $\bm{A}$ is defined as $\|\bm{A}\|$. The pseudo-inverse of a matrix $\bm{A}$ is defined as $\bm{A}^{\dag}$. For a function $h(\x):\R^p \to \R$, we write $\nabla_\x h(\x)|_{\x = \x^{*}}$ to denote the gradient or subdifferential of $h(\x)$, if they exist, evaluated at $\x=\x^{*}$. When the intention is clear, this may be written concisely as $\nabla h(\x^{*})$. We also follow the standard asymptotic notations. If for some constants $C, N >0$, $|f(n)|\le C|g(n)|$ for all $n\ge N$, then $f(n)=O(g(n))$; if $g(n)=O(f(n))$, then $f(n)=\Omega(g(n))$. Finally, Table~\ref{table:parameters} summarizes some key notations used in this paper for convenience.

\begin{table}[h]
  \footnotesize
  \begin{center}
    \begin{tabular}{@{}lll@{}}
      \toprule
      {\bf Symbol}  & {\bf Description} & {\bf Dimension}\\
      \midrule \addlinespace[1mm]
      $ \bm{\Delta}$ &  oriented incidence matrix &  $m \times n$\\ 
  $ \Dk$ &  $k$th order graph difference operator &  $r \times n$\\ 
      $ \betab$  & scalar-valued graph signal &  $n $\\ 
      $\B $ & vector-valued graph signal &  $n \times d$\\ 
      $\y $ & noisy observation of $\betab$ &  $n$\\ 
      $ \bm{Y} $ &  noisy observation of $\B$ &  $n \times d$\\
      $ \bm{\Delta}_{\ell \cdot}$ &  $\ell$-th row of $\bm{\Delta}$ &  $n $\\
	  $\B_{\cdot j} $ & $j$-th column of $\B$ &  $n$\\ 
	  $\|\Dk\| $ & spectral norm of $\Dk$ &  $1$\\ 

      \bottomrule
    \end{tabular}
  \end{center}
  \caption{\label{table:parameters}
    Key notations used in this paper. }
\end{table}

\section{Related Work and Connections} \label{sec:related_works}

Estimators that adapt to spatial inhomogeneities have been well studied in the literature via regularized regression, total variation and splines \cite{mammen_locally_1997}\nocite{rudin_nonlinear_1992}-\cite{chan_digital_2001}. Most of these methods involve locating change points or knots that denote a distinct change in the behavior of the function or the signal. 

Our work is most related to the spatially adaptive GTF estimator introduced in~\cite{wang_trend_2016} that~\emph{smoothens} or~\emph{filters} noisy signals to promote piecewise smooth behavior with respect to the underlying graph structure; see also \cite{mahmood_adaptive_2018}. In the same spirit as~\cite{mammen_locally_1997}, the fused LASSO and univariate trend filtering framework developed in~\cite{kim_ell_1_2009,tibshirani_adaptive_2014,tibshirani_sparsity_2005} use discrete difference operators to fit a time series signal using piecewise polynomials. The GTF framework generalizes univariate trend filtering by generalizing a path graph to arbitrarily complex graphs. Specifically, by appropriately defining the discrete difference operator, we can enforce piecewise constant, piecewise linear, and more generally piecewise polynomial behaviors over the graph structure. In comparison to previous work~\cite{wang_trend_2016}, in this paper, we have significantly expanded its scope by allowing vector-valued data over the graph nodes and a broader family of possibly non-convex penalties. 

We note that while a significant portion of the relevant literature on GTF or the fused LASSO has focused on the sparsistency or support recovery conditions under which we can ensure the recovery of the location of the discontinuities or knots~\cite{sharpnack_sparsistency_2012,harchaoui_multiple_2010}, in this work, we study the asymptotic error rates of our estimator with respect to the mean squared error. Our analysis of error rates leverages techniques in \cite{hutter_optimal_2016, dalalyan_prediction_2017} that result in sharp error rates of total variation denoising via oracle inequalities, which we have carefully adapted to allow {\em non-convex} regularizers. The obtained error rates can be translated into bounds on support recovery or how well we can localize the boundary by leveraging techniques in \cite{lin_approximate_2016}. 

Employing a graph-based regularizer that promotes similarities between the signal values at connected nodes has been investigated by many communities, such as graph signal processing, machine learning, applied mathematics, and network science. The Network LASSO proposed in \cite{hallac_network_2015}, which is similar to the GTF framework with multi-dimensional or vector-valued data, focused on the development of efficient algorithms without any theoretical guarantees. The recent works by Jung et al.~\cite{jung2018network,tran2017network,jung2017network} have analyzed the performance of Network LASSO for semi-supervised learning when the graph signal is assumed to be {\em clustered} according to the labels using the network null space property and the network compatibility condition inspired by related concepts in compressed sensing \cite{van2009conditions}. In contrast, our analysis does not make assumptions on the graph signal, and the error rate is adaptive to the alignment of the signal and the graph structure used in denoising.

A well-studied generalization of the sparse linear inverse problem is when there are multiple measurement vectors (MMV), and the solutions are assumed to have a~\emph{common sparsity pattern}~\cite{cotter_sparse_2005}\nocite{ chen_theoretical_2006}-\cite{eldar2010block}. Sharing information across measurements, and thereby exploiting the conformity of the sparsity pattern, has been shown to significantly improve the performance of sparse recovery in compressive sensing and sparse coding  \cite{mairal_online_2010}\nocite{chen_hyperspectral_2011, li_off--grid_2016,eldar_robust_2009}-\cite{davies_rank_2012}. Motivated by these works, we consider vector-valued graph signals that are regarded as multiple measurements of scalar-valued graph signals sharing discontinuity patterns. 

There are a few variants of non-convex penalties that promote sparsity such as SCAD, MCP, weakly convex penalties, and $\ell_q$ ($0\leq q<1$) minimization \cite{loh_regularized_2013,loh_support_2017}\nocite{chen2014convergence,chartrand2008restricted}-\cite{ji2018learning}. In this paper, we consider and develop theory for a family of non-convex penalties parametrized similarly to that in \cite{loh_regularized_2013,chen2014convergence} with SCAD and MCP as our prime examples, although it is valid for other non-convex penalties.

\begin{figure*}[t]
    \centering
    \includegraphics[width=0.95\textwidth]{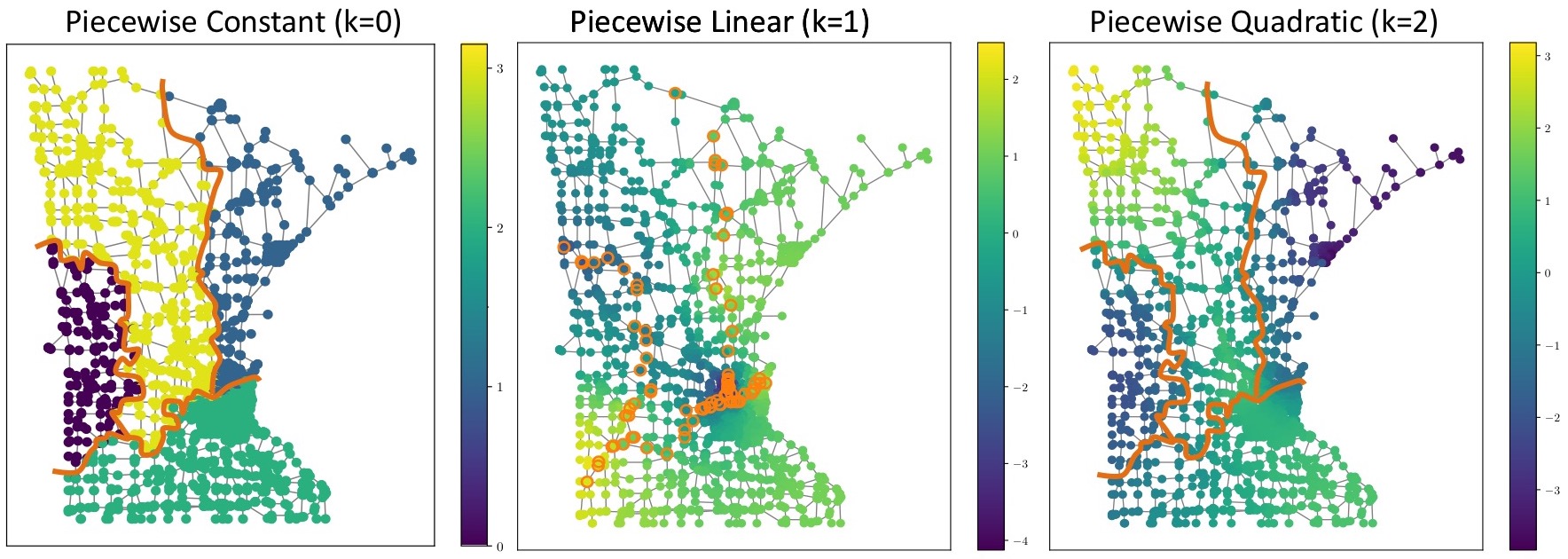}
    \caption{Illustration of piecewise smooth signals on the Minnesota road graph. From left to right: piecewise constant ($k=0$), piecewise linear ($k=1$), and piecewise quadratic ($k=2$) graph signals. Note that the highlighted change points, i.e. the support of $\Dk\betastar$, are edges for even $k$ and nodes for odd $k$.}
    \label{fig:nyc}
\end{figure*}

\section{Graph Signal Processing, Piecewise Smooth Signals, and Graph Trend Filtering} \label{sec:basics}

We consider an undirected graph $G = (\V, \E, \Adj)$, where $\V = \{v_1,\ldots, v_{n}\}$ is the set of nodes, $\E = \{e_1,\ldots, e_{m}\}$ is the set of edges, and $\Adj =[A_{jk}] \in \R^{n \times n}$ is the unweighted adjacency matrix -- also known as the~\emph{graph shift operator}~\cite{sandryhaila_discrete_2013}.  The edge set $\E$ represents the connections of the undirected graph $G$, and the positive edge weight $A_{jk}$ measures the underlying relation between the $j$th and the $k$th node, such as a similarity, a dependency, or a communication pattern.  
Let a scalar-valued~\emph{graph signal} be defined as
\begin{equation}
\label{eq:graph_signal}
\nonumber
  \betab \ = \ \begin{bmatrix}
 \beta_1, \beta_2, \ldots, \beta_{n}
\end{bmatrix}^{\top} \in \R^n,
\end{equation}
 where $\beta_i$ denotes the signal coefficient at the $i$th node. 
 
Let $\bm{\Delta} \in \R^{m \times n}$ be the oriented incidence matrix of $G$, where each row corresponds to an edge. That is, if the edge $e_i =(j,k) \in \mathcal{E}$ connects the $j$th node to the $k$th node ($j < k$), the entries in the $i$th row of $\bm{\Delta}$ is then given as
\begin{equation*}
\label{eq:Delta}
 \Delta_{i\ell} = 
  \left\{ 
    \begin{array}{rl}
      - 1, & \ell = j;\\
       1, & \ell = k;\\
      0, & \mbox{otherwise}
  \end{array} \right..
\end{equation*}

The entries of the signal $\bm{\Delta}\bm{\beta}=[(\beta_k -\beta_j)]_{(j,k)\in\mathcal{E}}$ specifies the unweighted pairwise differences of the graph signal over each edge. As a result, $\bm{\Delta}$ can be interpreted as a~\emph{graph difference operator}. In graph signal processing, a signal is called smooth over a graph $G$ if $\|\bm{\Delta}\bm{\beta}\|_2^2 = \sum_{(j,k) \in \mathcal{E}} ( \beta_k - \beta_j )^2$ is small.

\subsection{Piecewise Smooth Graph Signals}
In practice, the graph signal may not be necessarily smooth over the entire graph, but only locally within different pieces of the graph. To model inhomogeneous levels of smoothness over a graph, we say that a graph signal $\bm{\beta}$ is piecewise constant over a graph $G$ if many of the differences $\beta_k - \beta_j$ are zero for $(j,k)\in\mathcal{E}$. Consequently, the difference signal $\bm{\Delta}\bm{\beta}$ is sparse and $\| \bm{\Delta} \bm{\beta} \|_0$ is small.

We can characterize~\emph{piecewise $k$th order polynomial} signals on a graph, where the piecewise constant case corresponds to $k=0$, by generalizing the notion of graph difference operators. Specifically, we use the following recursive definition of the $k$th order graph difference operator $\Dk$ \cite{wang_trend_2016}. Let $\bm{\Delta}^{(1)} = \bm{\Delta}$ for $k=0$. For $k\geq 1$, let
\begin{equation} \label{eq:def_Dk}
	\bm{\Delta}^{(k+1)} = \begin{cases}
	\bm{\Delta}^{(1)\top} \bm{\Delta}^{(k)} \in \mathbb{R}^{n \times n} , & \text{ odd $k$ } \\ 
		\bm{\Delta}^{(1)} \bm{\Delta}^{(k)} \in \mathbb{R}^{m \times n} , & \text{ even $k$ } \\
	\end{cases}.
\end{equation}

The signal $\bm{\beta}$ is said to be a piecewise $k$th order polynomial graph signal if $\|\Dk\bm{\beta}\|_0$ is small. To further illustrate, let us consider the piecewise linear graph signal, corresponding to $k=1$, as a signal whose value at a node can be linearly interpolated from the weighted average of the values at neighboring nodes. It is easy to see that this is the same as requiring the second-order differences $\bm{\Delta}^{\top} \bm{\Delta} \bm{\beta}$ to be sparse. Similarly, we say that a signal has a piecewise quadratic structure over a graph if the differences between the second-order differences defined for piecewise linear signals are mostly zero, that is, if $\bm{\Delta} \bm{\Delta}^{\top} \bm{\Delta} \bm{\beta}$ is sparse. Fig.~\ref{fig:nyc} illustrates various orders of piecewise graph smooth signals over the Minnesota road network graph. 
 
\subsection{Denoising Piecewise Smooth Graph Signals via GTF} 
Assume we observe a noisy signal $\y$ over the graph under i.i.d Gaussian noise:
\begin{equation}
	\label{eq:model}
    \y = \betab^{\star} +\bm{\epsilon},\hspace{0.5cm} \bm{\epsilon} \sim \mathcal{N}(\bm{0}, \sigma^2 \bm{I}) ,
    \end{equation}
and seek to reconstruct $\betab^{\star}$ from $\y$ by leveraging the graph structure. When $\bm{\beta}$ is a smooth graph signal, Laplacian smoothing~\cite{belkin_regularization_2004,zhu_semi-supervised_2003,belkin_laplacian_2002}-\nocite{belkin_laplacian_2003}\cite{talukdar_new_2009} can be used, which solves the following problem: 
\begin{equation}
\label{eq:lapreg}
   \min_{\betab \in \R^n } \frac{1}{2}\|\y-\betab\|_2^2 + \lambda \| \bm{\Delta} \betab \|_2^2,
\end{equation}
where $\lambda>0$. However, it cannot localize abrupt changes in the graph signal when the signal is piecewise smooth.

Graph trend filtering (GTF)~\cite{wang_trend_2016} is a flexible framework for estimation on graphs that is adaptive to inhomogeneity in the level of smoothness of an
observed signal across nodes. The $k$th order GTF estimate is defined as:
\begin{equation} \label{eq:gtf_ell1}
    \min_{\betab \in \R^n } \frac{1}{2}\|\y-\betab\|_2^2 + \lambda \| \Dk \betab \|_1 ,
\end{equation}
which can be regarded as applying total variation or fused LASSO with the graph difference operator $\Dk$~\cite{kim_ell_1_2009,tibshirani_solution_2011}.  The sparsity-promoting properties of the $\ell_1$ norm have been well-studied~\cite{tibshirani_regression_1996}. Consequently, applying the $\ell_1$ penalty in GTF sets many of the (higher-order) graph differences to zero while keeping a small fraction of non-zero values. GTF is then~\emph{adaptive} over the graph; its estimate at a node adapts to the smoothness in its localized neighborhood. 

\section{Vector-Valued GTF with Non-convex Penalties} \label{sec:gtf_nonconvex}

In this section, we first extend GTF to allow a broader family of non-convex penalties, and then extend it to handle vector-valued signals over the graph.

\subsection{(Non-)convex Penalties}

The $\ell_1$ norm penalty considered in  \eqref{eq:gtf_ell1} is well-known to produce biased estimates \cite{zhang2008sparsity}, which motivates us to extend the GTF framework to a broader class of sparsity-promoting regularizers that are not necessarily convex. We wish to minimize the following generalized $k$th order GTF loss function:
\begin{equation}
\label{eq:gtf}
f(\betab) =\frac{1}{2}\|\y-\betab\|_2^2 + g( \Dk \betab ; \lambda, \gamma), \quad \betab \in \R^n ,
\end{equation}
where 
$$ g(\Dk \betab ) \triangleq g( \Dk \betab;\lambda, \gamma ) = \sum_{\ell = 1}^r \rho( (\Dk \betab)_{\ell} ;\lambda, \gamma)$$  is a regularizer defined as the sum of the penalty function $\rho(\cdot; \lambda, \gamma): \R \to \R$ applied element-wise to $\Dk \betab$. Here, $r=m$ for even $k$ and $r=n$ for odd $k$ to account for different dimensions of $\Dk$; see \eqref{eq:def_Dk}. We will refer to the GTF estimator that minimizes $f(\bm{\beta})$ as \textit{scalar-GTF}. 

Similarly to~\cite{loh_regularized_2013,zhang_general_2012,chen2014convergence}, we consider a family of penalty functions $\rho(\cdot; \lambda,\gamma)$ that satisfies the following assumptions.
\begin{assumptions}\label{assumptions_1}
Assume $\rho(\cdot; \lambda,\gamma)$ satisfies the following:
	\begin{enumerate}[label=(\alph*)]
		\item  $\rho(t ; \lambda, \gamma)$ satisfies $\rho(0 ; \lambda, \gamma) = 0$, is symmetric around $0$, and is non-decreasing on the real non-negative line.
		\item  For $t \ge 0$, the function $ t \mapsto \frac{\rho(t;\lambda, \gamma)}{t}$ is non-increasing in $t$. Also, $\rho(t ; \lambda, \gamma)$ is differentiable for all $t \neq 0$ and sub-differentiable at $t = 0$, with $\lim_{t\to0^{+}} \rho'(t; \lambda, \gamma)=\lambda$. This upper bounds $\rho(t ; \lambda,\gamma) \le \lambda | t| $. 
		\item There exists $\mu>0$ such that $ \rho(t; \lambda,\gamma) + \frac{\mu}{2}t^2$ is convex. 
	\end{enumerate}
\end{assumptions}

Many penalty functions satisfy these assumptions. Besides the $\ell_1$ penalty, the non-convex SCAD \cite{fan_variable_2001} penalty
\begin{align} \label{eq:scad}
\rho_{\mathrm{SCAD}} (t; \lambda, \gamma) &= \lambda \int_{0}^{|t|} \min \left(1,  \frac{ (\gamma - u/\lambda)_{+}}{\gamma-1}\right) du, \gamma\geq 2,
\end{align}
and the MCP \cite{zhang_nearly_2010}
\begin{align}\label{eq:mcp}	
\rho_{\mathrm{MCP}} (t; \lambda, \gamma) &= \lambda \int_{0}^{|t|} \left(1 - \frac{u}{\lambda \gamma}\right)_{+} du, \quad \gamma\geq 1 
\end{align}
also satisfy them. We note that Assumption~\ref{assumptions_1} (c) is satisfied for SCAD with $\mu \geq \mu_{\mathrm{SCAD}} = \frac{1}{\gamma - 1}$ and for MCP with $\mu \geq \mu_{\mathrm{MCP}}  = \frac{1}{\gamma}$. Fig.~\ref{fig:penfuns} illustrates the $\ell_1$, SCAD and MCP penalties for comparison. While the non-convexity means that in general, we may not always find the global optimum of $f(\betab)$, it often affords us many other advantages.
 SCAD and MCP both taper off to a constant value, and hence apply less shrinkage for higher values. As a result, they mitigate the bias effect while promoting sparsity. Further, they are smooth and differentiable for $t \ge 0$ and are both upper bounded by the $\ell_1$ penalty for all $t$. 
\begin{figure}[t]
    \centering
    \includegraphics[width=0.3\textwidth]{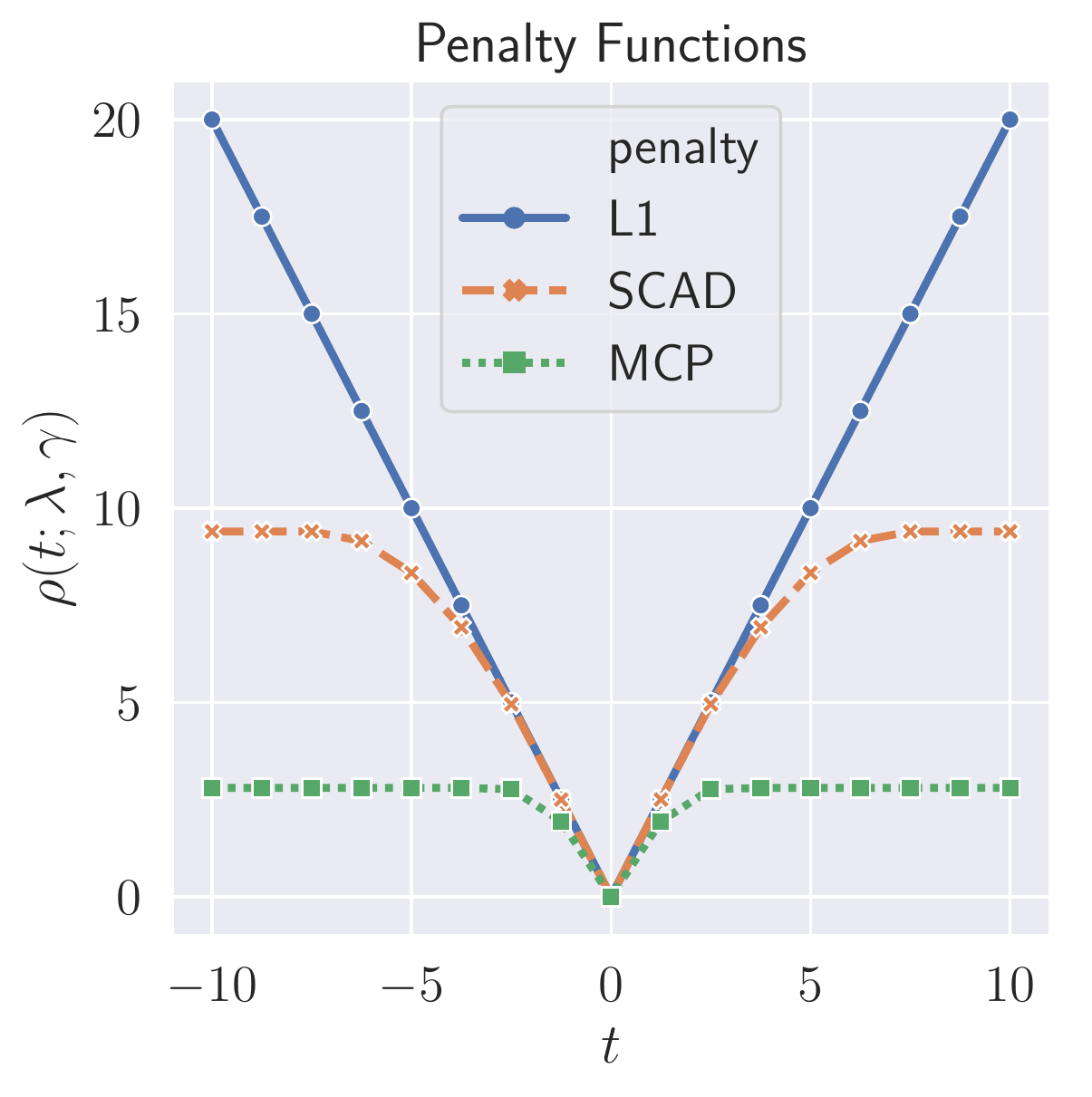}
    \caption{Illustration of $\rho(\cdot; \lambda, \gamma)$ for $\ell_1$, SCAD ($\gamma = 3.7$), and MCP ($\gamma = 1.4$), where $\lambda = 2$. Both SCAD and MCP move towards $\ell_1$ as $\gamma$ increases.}
    \label{fig:penfuns}
\end{figure}

\subsection{Vector-Valued GTF}
In many applications, the signals on each node are in fact multi-dimensional or~\emph{vector-valued}, e.g. time series in social networks, multi-class labels in semi-supervised learning, feature vectors of different objects in feature selection. Therefore, it is natural to consider an extension to the graph signal denoising problem, where the graph signal on each node is a $d$-dimensional vector instead of a scalar. In this scenario, we define a vector-valued graph signal to be piecewise smooth if it is piecewise smooth in each of its $d$ dimensions, and assume their discontinuities to coincide over the same small set of edges or nodes. Further, we denote the vector-valued signal of interest as $\B^{\star} \in \R^{n \times d}$, such that the $i$th row of the matrix $\B$ corresponds to the $i$th node of the graph. The noise model for the observation matrix $\bm{Y} \in \R^{n \times d}$ is defined as 
\begin{equation}
    \bm{Y} = \B^{\star} + \bm{E},
    \label{eq:modelvec}
\end{equation}
where each element of $\bm{E} \in \R^{n \times d}$ is drawn i.i.d from $\mathcal{N}(0, \sigma^2)$. A na\"{i}ve approach is to estimate each column $ \B_{\cdot j}$ of $\B$ separately via scalar-GTF:
\begin{equation}\label{eq:sep_gtf}
\min_{\B\in \R^{n \times d}}  \sum_{j=1}^d f(\B_{\cdot j}).
\end{equation}
 
However, this formulation does not take full advantage of the multi-dimensionality of the graph signal. Instead, when the columns of $\B$ are correlated, coupling them can be beneficial such that we encourage the sharing of information across dimensions or features. For example, if one column $\B_{\cdot i}$ exhibits strong piecewise smoothness over the graph, and therefore has compelling evidence about the relationship between nodes, sharing that information to a related column $\B_{\cdot j}$ can improve the overall denoising and filtering performance. As a result, we formulate a \textit{vector-GTF} problem as follows:
\begin{equation} \label{eq:gtfvec}
    \min_{\B \in \R^{n\times d} } \frac{1}{2}\|\bm{Y}-\B \|_{\mathrm{F}}^2 + h( \Dk \B; \lambda, \gamma),
\end{equation}
where the new penalty function $h(\Dk\B)\triangleq  h( \Dk\B ; \lambda, \gamma) : \R^{r \times d} \to \R$ is the sum of $\rho(\cdot; \lambda,\gamma)$ applied to the $\ell_2$ norm of each row of $\Dk\B \in \R^{r \times d}$:
 \begin{equation} \label{eq:gtfvecpen}
    h( \Dk\B \ ; \lambda, \gamma) = \sum_{\ell=1}^r\rho\left ( \| (\bm{\Delta}^{(k+1)}\B)_{\ell \cdot}  \|_2 ; \lambda, \gamma \right ).
 \end{equation}
By enforcing sparsity on $\left\{\| (\bm{\Delta}^{(k+1)}\B)_{\ell \cdot} \|_2 \right\}_{1 \leq l \leq r}$, we are coupling $\bm{\Delta}^{(k+1)} {\B}_{\cdot j}$ to be of similar sparsity patterns across $j=1,\ldots,d$. Note the difference from \eqref{eq:sep_gtf}, where elements of $(\bm{\Delta}^{(k+1)}\B)_{\ell \cdot}$ can be set to zero or non-zero independently.

\section{Theoretical Guarantees}
\label{sec:theory}

In this section, we present the error rates and support recovery guarantees of the generalized GTF estimators, namely  scalar-GTF~\eqref{eq:gtf} and vector-GTF~\eqref{eq:gtfvec}, under the AWGN noise model. Before continuing, we first define a few useful quantities. Let $C_{G}$ be the number of connected components in the graph $G$, or equivalently, the dimension of the null space of $\Dk$. Further, let $r$ be the number of rows of $\Dk$, and  $\zeta_k$ be the maximum $\ell_2$ norm of the columns of $\Dkdag$.

\subsection{Error Rates of First-order Stationary Points}

Due to non-convexity, global minima of the proposed GTF estimators may not be attainable. Therefore, it is more desirable to understand the statistical performance of any first-order stationary points of the GTF estimators. We call $\betahat \in \R^n$ a stationary point of $f(\betab)$, if it satisfies
\begin{equation}
    0 \in \nabla_\betab f(\betab)|_{\betab = \betahat} .\nonumber
\end{equation}

We further introduce the compatibility factor, which  generalizes the notion used in \cite{hutter_optimal_2016} to allow vector-valued  signals.
\begin{defn}[Compatibility factor]
Let $\Dk$ be fixed. The compatibility factor $\kappa_{T,d}$ of a set $T\subseteq \{1,2,\ldots, r\}$ is defined as $\kappa_{\varnothing, d}=1$, and for nonempty set $T$,      
\begin{equation}  \label{def:kappa}
\kappa_{T,d}(\Dk) = \inf_{\B \in \R^{n\times d}}\Bigg \{\frac{\sqrt{|T|} \cdot \|\B\|_{\mathrm{F}}}{\sum_{\ell\in T}\|(\Dk\B)_{\ell \cdot}\|_2} \Bigg \}. \nonumber
\end{equation}
\end{defn}

To further build intuition, consider $\sqrt{|T|}\kappa_{T,1}(\bm{\Delta})^{-1}=\sup_{\betab \in \R}\{\|(\bm{\Delta})_{T}\betab\|_1/\|\betab\|_2\}$. This is precisely the definition of $\|(\bm{\Delta})_{T}\|_{1,2}$, an induced norm of the $|T|\times n$ submatrix of $\bm{\Delta}$. If we consider signals with fixed power $\|\betab\|_2^2=1$, $\|(\bm{\Delta})_{T}\betab\|_1$ will depend on how much the $T$ edges are connected to each other. Together with $\|(\bm{\Delta})_{T}\betab\|_1\le \sqrt{|T|}\|(\bm{\Delta})_{T}\betab\|_2$, $\kappa_{T,1}(\bm{\Delta})$ can be related to the restricted eigenvalue condition, which is often used to bound the performance of LASSO \cite{van2009conditions}. With slight abuse of notation, we write $\kappa_{T}: = \kappa_{T,d}$. 

We have the following oracle inequality that is applicable to the stationary points of the GTF estimators. The proof in Appendix A follows a construction that is similar to Theorem 2 in~\cite{hutter_optimal_2016}. The oracle inequality holds for any $\betahat$ that satisfies the first-order optimality condition, allowing the use of non-convex penalties. This mild condition on $\betahat$ is a key difference from~\cite[Theorem 3]{wang_trend_2016} and~\cite[Theorem 1]{varma2019improving} that are applicable to the global minimizer, which is difficult to guarantee when using non-convex penalties. We also stress that although GTF was motivated by piecewise smooth graph signals, Theorem~\ref{thm:oracle} holds for any graph $G$ and graph signal $\betastar$. 
\begin{myThm}[Oracle inequality of GTF stationary points]\label{thm:error_oracle}
Assume $\mu < 1/\|\Dk\|^2$. Fix $\delta\in (0, 1)$. For scalar-GTF \eqref{eq:gtf}, let $\betahat$ be a stationary point. 
Set $\lambda = \sigma \zeta_k \sqrt{2 \log \left(\frac{er}{\delta}\right)} $,  then 
\begin{align}
&	\frac{ \| \betahat-\betastar  \|_2^2}{n}  \leq \inf_{\betab \in \R^{n}} \Bigg\{ \frac{\| \betab -\betastar \|_2^2 + 4 \glam {(\Dk \betab)_{T^c}}}{n}  \Bigg\}\nonumber \\
	&+ \frac{2\sigma^2\left[  C_G + 2 \sqrt{2 C_G \log(\frac{1}{\delta})}  + \frac{8 \zeta_k^2 |T|}{\kappa_{T}^2} \log (\frac{er}{\delta}) \right]}{n (1 -  \mu  \|\Dk\|^2 )}  \label{eq:scalar_oracle}
\end{align}
with probability at least $1-2\delta$ for any $T \subseteq \{1, 2, ..., r\}$. Similarly, for vector-GTF \eqref{eq:gtfvec}, let $\Bhat$ be a stationary point. Set $\lambda = \sigma \zeta_k \sqrt{2d \log (\frac{edr}{\delta})}$, then
\begin{align}
&	\frac{\| \Bhat - \B^\star \|_\mathrm{F}^2}{dn} \leq \inf_{\B \in \R^{n\times d}} \Bigg \{ \frac{ \| \B -\B^\star \|_\mathrm{F}^2 +4\hlam { (\Dk \B)_{T^c}} }{dn }\Bigg \} \nonumber \\
	&+ \frac{2\sigma^2\left[   C_G + 2 \sqrt{2 C_G \log(\frac{d}{\delta})}  + \frac{8 \zeta_k^2 |T|}{\kappa_{T}^2} \log (\frac{edr}{\delta}) \right]}{ n(1 -  \mu  \| \Dk \|^2)}  \label{eq:vector_oracle}
\end{align}
with probability at least $1-2\delta$ for any $T \subseteq \{1, 2, ..., r\}$.
\label{thm:oracle}
\end{myThm}%
 
 \begin{remark} \label{remark:mu}
Recall that $\mu$ is defined in Assumption~\ref{assumptions_1} (c), which characterizes how ``non-convex'' the regularizer is, and dictates the inflection point in Fig.~\ref{fig:penfuns}. The assumption $\mu < 1/\|\Dk\|^2$ in Theorem~\ref{thm:error_oracle} therefore implicitly constrains the level of non-convexity of the regularizer. Take MCP in \eqref{eq:mcp} for example: since $\mu \geq 1/\gamma $,  we can guarantee the existence of a valid $\mu$ such that $\mu < 1/\|\Dk\|^2$ as long as we set $\gamma > \|\Dk\|^2$.
\end{remark}

Theorem~\ref{thm:error_oracle} allows one to select $\betab$ and $T$ to optimize the error bounds on the right hand side of \eqref{eq:scalar_oracle} and \eqref{eq:vector_oracle}. For example, pick $\betab=\betastar$ in \eqref{eq:scalar_oracle} (hence an ``oracle'') to have
\begin{align} 	\label{eq:our_bound}
 \frac{ \| \betahat-\betastar  \|_2^2}{n} & \leq  \frac{  4 \glam {(\Dk \betastar)_{T^c}}}{n}   \nonumber \\
	&\quad + \frac{2\sigma^2\left[   C_{G}^{\delta} + 8 \zeta_k^2\kappa_{T}^{-2} |T| \log (\frac{er}{\delta}) \right]}{n (1 -  \mu  \| \Dk \|^2)},
\end{align}
where $C_{G}^{\delta} =C_G + 2 \sqrt{2 C_G \log(\frac{1}{\delta})}$.

\begin{itemize}
\item By setting $T$ as an empty set, we have 
\begin{align} 	
 \hspace{-0.6cm}\frac{ \| \betahat-\betastar  \|_2^2}{n} & \leq  \frac{  4 \glam {\Dk \betastar}}{n} + \frac{2\sigma^2C_{G}^{\delta}}{n (1 -  \mu  \| \Dk \|^2)}, \label{eq: scalar_oracle_emptyT}
\end{align}
which suggest that the reconstruction accuracy improves when the ground truth $\bm{\beta}^{\star}$ is better aligned with the graph structure, and consequently the value of $g(\Dk \betab^{\star})$ is small. 

\item On the other hand, by setting $T$ as the support of $\Dk\betastar$, we achieve
\begin{align*} 	
 \frac{ \| \betahat-\betastar  \|_2^2}{n} & \leq   \frac{2\sigma^2\left[   C_{G}^{\delta}  + 8 \zeta_k^2\kappa_{T}^{-2} \|\Dk\betastar\|_0  \log (\frac{er}{\delta}) \right]}{n (1 -  \mu  \| \Dk \|^2)},
\end{align*}
which grows linearly as we increase the sparsity level $\|\Dk\betastar\|_0$. 
\end{itemize}

Similar discussions can be conducted for vector-GTF by choosing $\B= \B^\star$ in \eqref{eq:vector_oracle}. More importantly, we can directly compare the performance of vector-GTF with scalar-GTF, which was formulated for vector-valued graph signals in \eqref{eq:sep_gtf}. The error bound of vector-GTF pays a small price in the order of $\log d$, but is tighter than scalar-GTF if $ h((\Dk\B^\star)_{T^c}) \ll\sum_{j=1}^d g((\Dk\B^\star_{\cdot j})_{T^c})$. This suggests that vector-GTF is much more advantageous when the support sets of $\Dk \B^\star_{\cdot j}$ for $j=1, \ldots, d$ overlap, i.e. when the local discontinuities and patterns in $\B^\star_{\cdot j}$ are shared.

\subsection{Comparison with Scalar-GTF using $\ell_1$ Regularization}

We compare our error bound for scalar-GTF that is on $ \| \betahat-\betastar  \|_2^2/n$ with \cite[Theorem 3]{wang_trend_2016}, which is obtained for GTF with the $\ell_1$ penalty, reproduced below for convenience.
\begin{myThm}[Basic error bound of $\ell_1$ GTF minimizer]
If $\lambda=\Theta(\sigma\zeta_k\sqrt{\log r})$, then $\betahat$, the minimizer of \eqref{eq:gtf_ell1}, satisfies
\begin{align*}
&	\frac{ \| \betahat-\betastar  \|_2^2}{n}  = O\left(  \frac{\lambda \|\Dk \betastar\|_1}{n} + \frac{\sigma^2 C_G}{n} \right). 
\end{align*}
\label{thm:l1-gtf}
\end{myThm}
The above bound is comparable to our bound in the special case of setting $T$ to an empty set, i.e. \eqref{eq: scalar_oracle_emptyT}. The first term of the bound in \eqref{eq: scalar_oracle_emptyT} is upper bounded by that of Theorem~\ref{thm:l1-gtf}. The non-convex regularization yields especially tighter bounds when $\Dk \betab^{\star}$ contains large coefficients, so that $ g(\Dk \betab^{\star}) \ll \lambda\| \Dk \betab^{\star} \|_1 $. On the other hand, the second term of \eqref{eq: scalar_oracle_emptyT} contains $1-\mu\|\Dk\|^2$ in the denominator, which makes it an upper bound of the second term in Theorem~\ref{thm:l1-gtf}. This gap can be brought down by choosing a larger $\gamma$, which allows one to pick a smaller $\mu$, as mentioned in Remark~\ref{remark:mu}. However, as $\gamma \to \infty$, non-convex SCAD and MCP also tends to $\ell_1$, which erases the improvement from using non-convex regularizers in the first term of the bound. This indicates a trade-off in the overall error bound based on $\gamma$, or the ``non-convexity" of the regularizers chosen for scalar-GTF. 
 
To sum up, despite being non-convex, we can guarantee that any stationary point of the proposed GTF estimator possesses strong statistical guarantees.

\subsection{Error Rates for Erd\H{o}s-R\'enyi Graphs}
We next specialize Theorem~\ref{thm:error_oracle} to the Erd\H{o}s-R\'enyi random graphs using spectral graph theory \cite{chung_spectra_2011}. Let $d_{\max}$ and $d_0$ respectively be the maximum and expected degree of the graph. It is known that for any graph it holds \cite{wang_trend_2016} 
\begin{equation} \label{eq:zeta_bound}
\zeta_k \le \lambda_{\min}(\bm{\Delta}^{(2)})^{-\frac{k+1}{2}},
\end{equation} 
where $\lambda_{\min}(\bm{\Delta}^{(2)})$ is the smallest {\em non-zero} eigenvalue of the graph Laplacian matrix $\bm{\Delta}^{(2)}$. Moreover, we have $ \| \Dk \|^2 = (\lambda_{\max}( \bm{\Delta}^{(2)}) )^{k+1}$, and  $d_{\max} +1\le \lambda_{\max}(\bm{\Delta}^{(2)})  \leq  2d_{\max} $~\cite{chung_spectral_1997}.  Next, we present a simple lower bound on $\kappa_T$, which is proved in Appendix~\ref{proof_prop_kappa}.
\begin{myProposition}[Bound on $\kappa_T$] \label{prop:kappa}
 $\kappa_T$ is bounded for any $ T $ and $d$ as
\begin{equation}
    \kappa_T(\Dk) \ge (2d_{\max})^{-\frac{k+1}{2}}. \nonumber
\end{equation}
\end{myProposition}

For an Erd\H{o}s-R\'enyi random graph, if $d_0=\Omega\left(\log(n)\right)$, we have $d_{\max}=O(d_0)$ almost surely \cite[Corollary 8.2]{blum2016foundations} and $C_G=1$. Furthermore, $\lambda_{\min}(\bm{\Delta}^{(2)})=\Omega(d_0-\sqrt{d_0})$ \cite{wang_trend_2016,chung_spectra_2011,lubotzky_ramanujan_1988}, and $r = n$ for odd $k$ and $r=O(nd_0)$ for even $k$. Therefore, with probability at least $1-n^{-10}$, we have
\begin{align*}
  \frac{\| \betahat-\betastar  \|_2^2}{n} &\lesssim \frac{ \sigma^2\sqrt{\log n}}{n}   \\
&  +  \min\left\{ \frac{  \glam {\Dk \betastar}}{n} ,\frac{ \sigma^2 \|\Dk \betastar\|_0 \log n   }{n}  \right\} ,
\end{align*}
where $  \glam {\Dk \betastar}   \lesssim \frac{ \sigma \|\Dk \betastar\|_1 \sqrt{\log n} }{ d_0^{(k+1)/2}}$ by plugging in $\glam {\Dk \betastar} \leq \lambda \|\Dk \betastar\|_1$.

These results are also applicable to $d_0$-regular Ramanujan graphs \cite{lubotzky_ramanujan_1988}.

\subsection{Support Recovery}
An alternative yet important metric for gauging the success of the proposed GTF estimators is support recovery, which aims to localize the discontinuities in the piecewise smooth graph signals, i.e. the support set of $\Dk \betab^{\star}$, that is 
\begin{equation}
	S_k(\betab^{\star}) = \left \{ t \in \{ 1,\cdots,r \}: (\Dk \betab^{\star})_t \neq 0  \right \} .
\end{equation}
In particular, for odd $k$, the discontinuities correspond to graph nodes; and for even $k$, they correspond to the edges. Let $\betahat$ be the GTF estimate of the graph signal. The quality of the support recovery can be measured using the~\emph{graph screening distance} \cite{lin_approximate_2016}. For any $t_1 \in S_k(\betab^{\star}) $ and $t_2 \in S_k(\betahat)$, let $d_G(t_1,t_2)$ denote the length of the shortest path between them. The distance of $S_k(\betahat)$ from $S_k(\betab^{\star}) $ is then defined as
\begin{align}  
	&d_G (S_k(\betahat)| S_k(\betastar)) \nonumber \\
	&= 
	\begin{cases}
		\displaystyle\max_{t_1 \in S_k(\betastar)}\min_{t_2 \in S_k(\betahat)} d_G(t_1,t_2), &\text{if}~S_k(\betastar) \neq \emptyset\\
		\quad\infty &\text{otherwise} 
	\end{cases} .
\end{align}

Interestingly, Lin et.al. \cite{lin_approximate_2016} showed recently that under mild assumptions,
one can translate the error bound into a support recovery guarantee. Specifically, letting $R_n$ be the RHS of \eqref{eq:scalar_oracle} that bounds the error $\|\betahat -\betastar\|_2^2/n$ in Theorem~\ref{thm:error_oracle}, we have
\begin{equation} \label{eq:support_recovery}
	d_G(S_k(\betahat)| S_k(\betastar)) = \begin{cases}
		O\left(\frac{ R_n}{H_r^2} \right) 	 , & k=0 \\ 
		O\left( \frac{R_n^{1/3}}{H_r^{2/3}} \right)  	 , &  k =1 \\
	\end{cases},
\end{equation}
where $H_r$ quantifies the minimum level of discontinuity, defined as the minimum absolute value of the non-zero values of $\Dk \betab^{\star}$, i.e.
\begin{equation}
	H_r = \min_{t \in S_k(\betastar)}  | (\Dk \betab^{\star})_t |.
\end{equation}
Consequently, this leads to support recovery guarantees of the proposed GTF estimators. Numerical experiment in Section~\ref{sec:simulation_denoising} verifies the superior performance of the non-convex regularizers over the $\ell_1$ regularizer for support recovery.

\section{ADMM Algorithm and its Convergence}\label{sec:admm}
There are many algorithmic approaches to optimize the vector-GTF formulation in \eqref{eq:gtfvec}, since scalar-GTF \eqref{eq:gtf} can be regarded as a special case with $d=1$. In this section, we illustrate the approach adopted in this paper, which is the Alternating Direction Method of Multipliers (ADMM) framework for solving separable optimization problems~\cite{boyd_distributed_2011}. 

Via a change of variable as $\bm{Z} =\Dk\B$, we can transform \eqref{eq:gtfvec} to 
\begin{align*}
    \min_{\B \in \R^{n\times d} } \frac{1}{2}\|\bm{Y}-\B \|_{\mathrm{F}}^2 + h( \bm{Z}; \lambda, \gamma) \quad \text{   s.t.  }\;   \bm{Z}=\Dk\B.
  \end{align*}
Its corresponding Lagrangian can be written as:
  \begin{align}
    \mathcal{L}(\B, \bm{Z}, \bm{U}) &= \frac{1}{2}\|\bm{Y}-\B\|_{\mathrm{F}}^2+ h( \bm{Z}; \lambda, \gamma) \nonumber \\
    &+ \frac{\tau}{2}\|\Dk\B-\bm{Z}+\bm{U}\|_{\mathrm{F}}^2  -\frac{\tau}{2}\|\bm{U}\|_{\mathrm{F}}^2  ,\label{eq:nc-gtf-lag-vec}
\end{align}
where $\bm{U} \in \R^{r \times d}$ is the Lagrangian multiplier, and $\tau$ is the parameter. Alg.~\ref{alg:admmvec} shows the ADMM updates based on the Lagrangian in \eqref{eq:nc-gtf-lag-vec}. Recall the proximal operator is defined
as $\texttt{Prox}_{f}(\bm{v}; \alpha) = \argmin_\x \frac{1}{2}\|\x-\bm{v}\|_2^2 + \alpha f(\x) $ for a function $f(\cdot)$. $\ell_1$, SCAD and MCP all admit closed-form solutions of $\texttt{Prox}$, which are simple thresholding operations \cite{huang_selective_2012}. Furthermore, we have the following convergence guarantee for Alg.~\ref{alg:admmvec}, whose proof is provided in Appendix~\ref{proof:convergence_admm}. Theorem~\ref{thm:conv} implies that the output of Alg.~\ref{alg:admmvec} satisfies Theorem~\ref{thm:oracle}.
\begin{myThm}
	\label{thm:conv}
	Let $\tau \geq \mu$, then Alg.~\ref{alg:admmvec} converges to a stationary point of \eqref{eq:gtfvec}.
\end{myThm}

\begin{algorithm}[t]
\caption{ADMM for solving \eqref{eq:gtfvec}}\label{alg:admmvec}
\begin{algorithmic}[1]
\State \textbf{Inputs:} $\bm{Y}, \Dk$, and parameters $\lambda$, $\gamma$, $\tau$
    \Initialize{     $\B \gets \bm{Y}$ or $\B_\texttt{init}$ if given.\\
    $\D \gets\Dk$, $\bm{Z} \gets \D\B$, $\bm{U} \gets \D\B - \bm{Z}$\\
    $\bm{X} \gets (\bm{I} +\tau\D^{\top}\D)^{-1}$
}
\Repeat
        \For{\texttt{$j \gets$ 1 to num\_cols($\B$)}}
        \State $\B_{\cdot j} \gets \bm{X}(\tau\D^{\top}(\bm{Z}_{\cdot j}-\bm{U}_{\cdot j})+\bm{Y}_{\cdot j})$
        \EndFor 
        \For{\texttt{$\ell \gets$ 1 to num\_rows($\D\B$)}}
        \State {$\bm{Z}_{\ell \cdot} \gets \texttt{Prox}_{\rho} (\|\D_{\ell \cdot} \B + \bm{U}_{\ell \cdot}\|_2; \lambda/\tau) $}
        \EndFor 
        \State $\bm{U}\gets  \bm{U} + \D \B-\bm{Z}$
\Until{termination}
\end{algorithmic}
\end{algorithm}
In addition, we provide a detailed time complexity analysis of Alg.~\ref{alg:admmvec} in Table~\ref{table:complexity}. Note that since $\bm{\Delta}$ is a sparse matrix with exactly $2m$ non-zero entries, Alg.~\ref{alg:admmvec} can run much faster when $k=0$. As a preprocessing step for each $\D$, we compute $\bm{V} \in \R^{n\times n}$ and $\bm{S} \in \R^{n \times n}$, the eigenvectors and eigenvalues of $\D^{\top}\D$, exactly once. $\bm{X}=\bm{V}(\bm{I}+\tau \bm{S})^{-1}\bm{V}^\top$ can then be initialized very efficiently for all experiments that use $\D$. 
\begin{table}[h]
 \footnotesize
  \begin{center}
    \begin{tabular}{@{}rcc@{}}
      \toprule
        &     \normalsize  \bf{$k \ge 1$}  &     \normalsize \bf{$k=0$}  \\    \midrule \addlinespace[2mm]  
$\D^\top\D$ eigen decomposition & $O(rn^2+n^3)$ &  $O(m^2+n^3)$ \\\midrule \addlinespace[1mm] 
$\Z$ initialization & $O(rnd)$ &  $O(md)$ \\\addlinespace[1mm]
$\bm{X}$ initialization  & $O(n^2)$ & $O(n^2)$ \\\addlinespace[1mm]
$\B$ update  & $O(d(nr+n^2))$   & $O(d(m+n^2))$  \\\addlinespace[1mm]
$\D\B$ calculation  & $O(rnd)$ &   $O(md)$ \\\addlinespace[1mm]
$\Z, \U$ update &  $O(rd)$     &   $O(rd)$           \\\addlinespace[2mm]
Total after $t$ iterations & $O(tdrn+tdn^2)$      & $O(tdm+tdn^2)$        \\
      \bottomrule
    \end{tabular}
  \end{center}
  \caption{\label{table:complexity}
  Time complexity analysis of Alg.~\ref{alg:admmvec}.
 }
\end{table}

\section{Numerical Experiments}\label{sec:numerical}
For the following experiments, we fixed $\gamma=3.7$ for SCAD, $\gamma=1.4$ for MCP. Further, the graphs we use in the following experiments satisfy Assumption 2 for this choice of $\gamma$. Unless explicitly mentioned, we tuned $\lambda$ and $\frac{\tau}{\lambda}$ for each experiment using the Hyperopt toolbox \cite{bergstra_making_2013}. 
To meet the convergence criteria in Theorem~\ref{thm:conv}, we enforced $\tau \ge 1/\gamma$. SCAD/MCP were warm-started with the GTF estimate with $\ell_1$ penalty. Python packages PyGSP \cite{defferrard_pygsp:_nodate} and NetworkX \cite{hagberg_exploring_2008} were used to construct and plot graphs. The input signal SNR was calculated as $ 10\log_{10}(\|\Bstar\|_\mathrm{F}/\sigma^2nd)$, while the reconstructed signal SNR was calculated as $10\log_{10}(\|\Bstar\|_\mathrm{F}/\|\Bhat-\Bstar\|_\mathrm{F})$, where $\Bhat$ was the reconstruction. Computation time was measured with MacBook Pro 2017 with an 2.9 GHz Intel Core i7 and 16GB RAM. Our code is available at \url{https://github.com/HarlinLee/nonconvex-GTF-public}.

\subsection{Denoising via GTF with Non-convex Regularizers}\label{sec:simulation_denoising}

\begin{figure}[t]
    \centering
    \includegraphics[width=\linewidth]{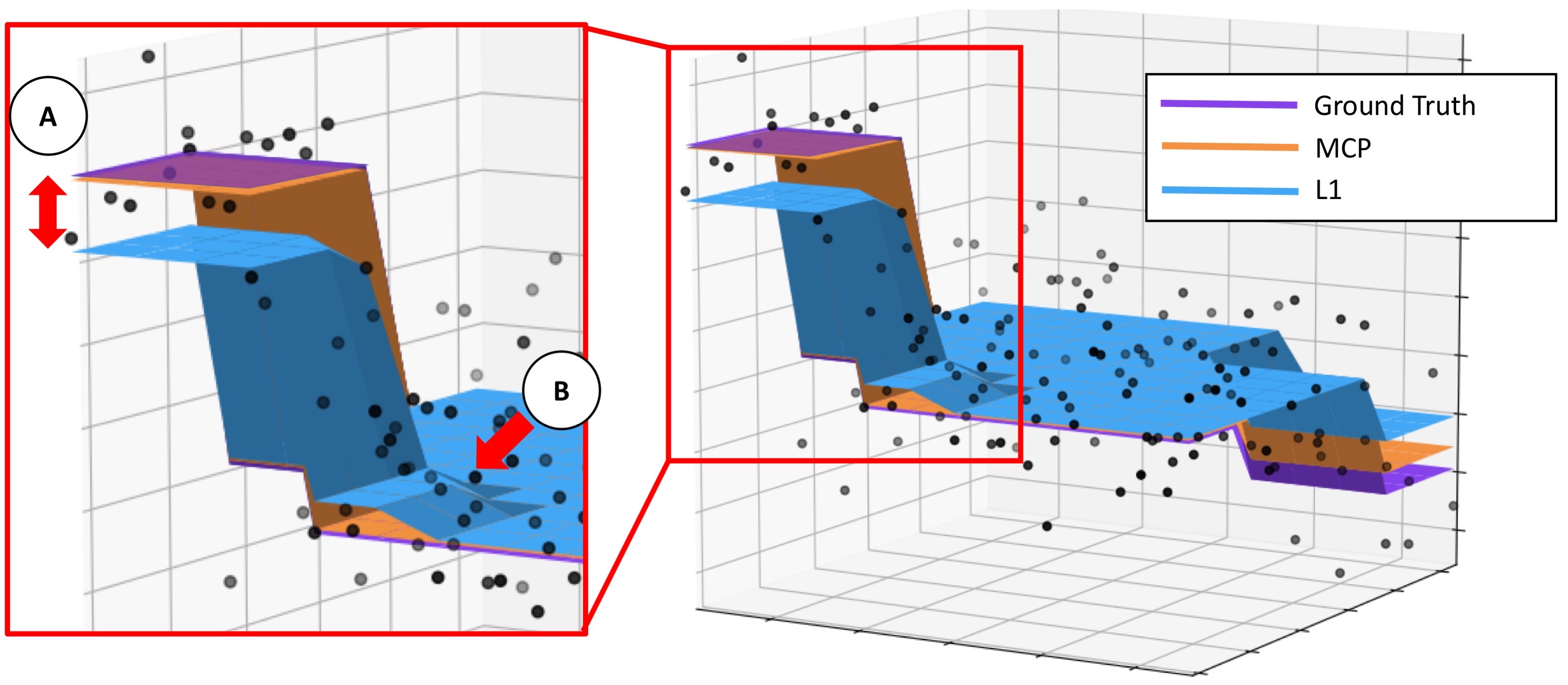}
    \caption{Scalar-GTF with MCP (orange) has much lower bias than scalar-GTF with $\ell_1$ (blue) when estimating a piecewise constant signal over a $12\times 12$ grid graph. See highlighted regions pointed by red arrows in A and B. The scatter points correspond to a noisy signal with $5$dB SNR.  }\label{fig:bias}
\end{figure}
\begin{figure}[t]
    \centering
    \includegraphics[width=0.35 \textwidth]{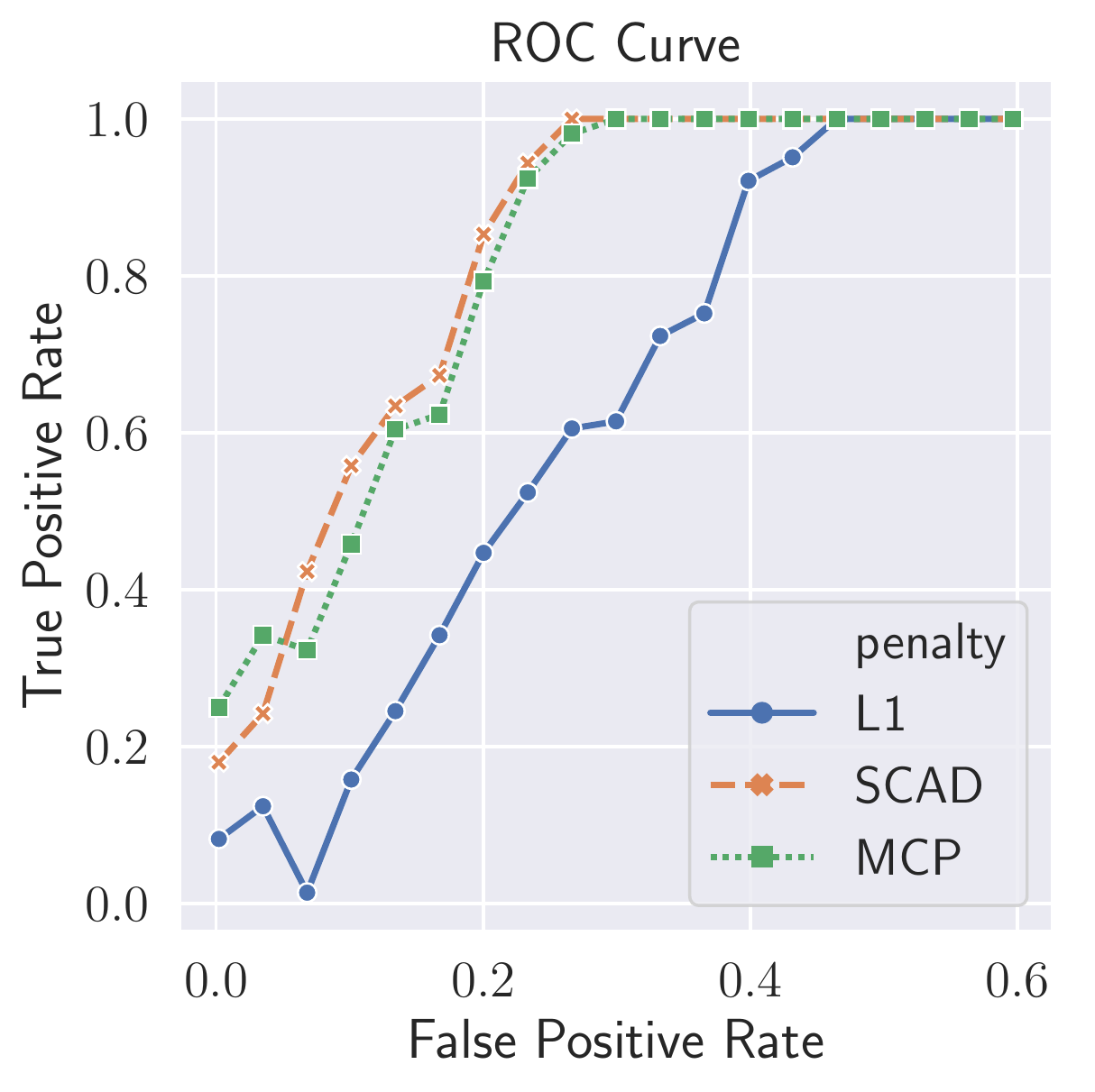}
    \caption{The ROC curve for classifying whether an edge lies on a boundary for the Minnesota road graph signal shown in Fig.~\ref{fig:snr}. The input SNR of the noisy piecewise constant signal is $7.8$dB.  }
    \label{fig:dfsweep}
\end{figure}
We first highlight via synthetic examples two important advantages that non-convex regularizers provide over the $\ell_1$ penalty.

\begin{itemize}
\item \textbf{Bias Reduction:} We demonstrate the reduction in signal bias in Fig.~\ref{fig:bias} for the graph signal defined over a $12\times 12$ 2D-grid graph, using both the $\ell_1$ penalty and the MCP penalty. Clearly, the MCP estimate (orange) has less bias than the $\ell_1$ estimate (blue), and can recover the ground truth surface (purple) more closely.

\item\textbf{Support Recovery:} We illustrate the improved support recovery performance of non-convex penalties on localizing the boundaries for a piecewise constant signal on the Minnesota road graph, shown in Fig.~\ref{fig:snr}. Particularly, we look at how well our estimator localizes the support of $\Dk \betastar$, that is, the discontinuity of the piecewise constant graph signal by looking at how well we can classify an edge as connecting two nodes in the same piece or being a cut edge across two pieces. By sweeping the regularization parameter $\lambda$, we obtain the ROC curve in Fig.~\ref{fig:dfsweep}, i.e. the true positive rate versus the false positive rate of classifying a cut edge correctly, and see that scalar-GTF with MCP and SCAD consistently outperforms the scalar-GTF with $\ell_1$ penalty.
\end{itemize}

\begin{figure*}[th]
\begin{minipage}[b]{0.32\linewidth}
\centering
    \includegraphics[width=\linewidth]{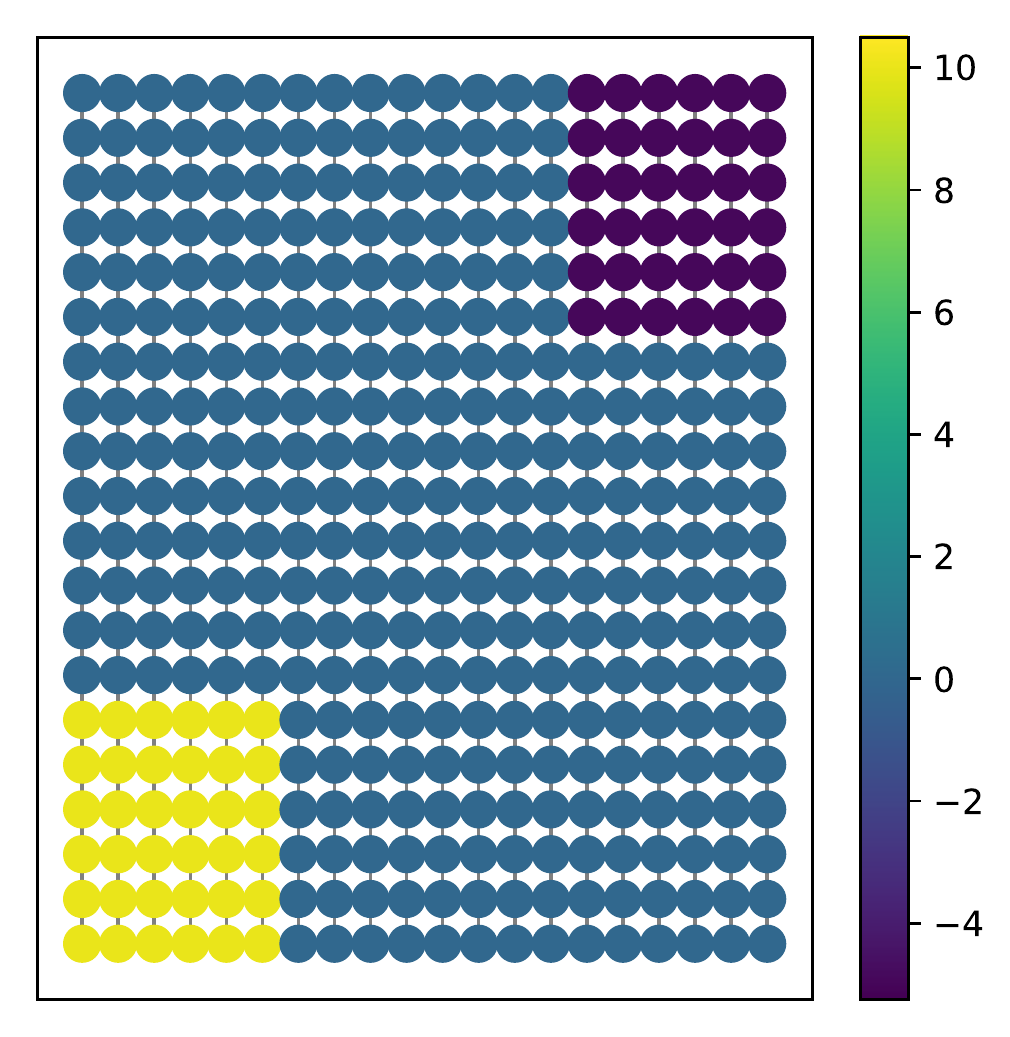}
    \includegraphics[width=\linewidth]{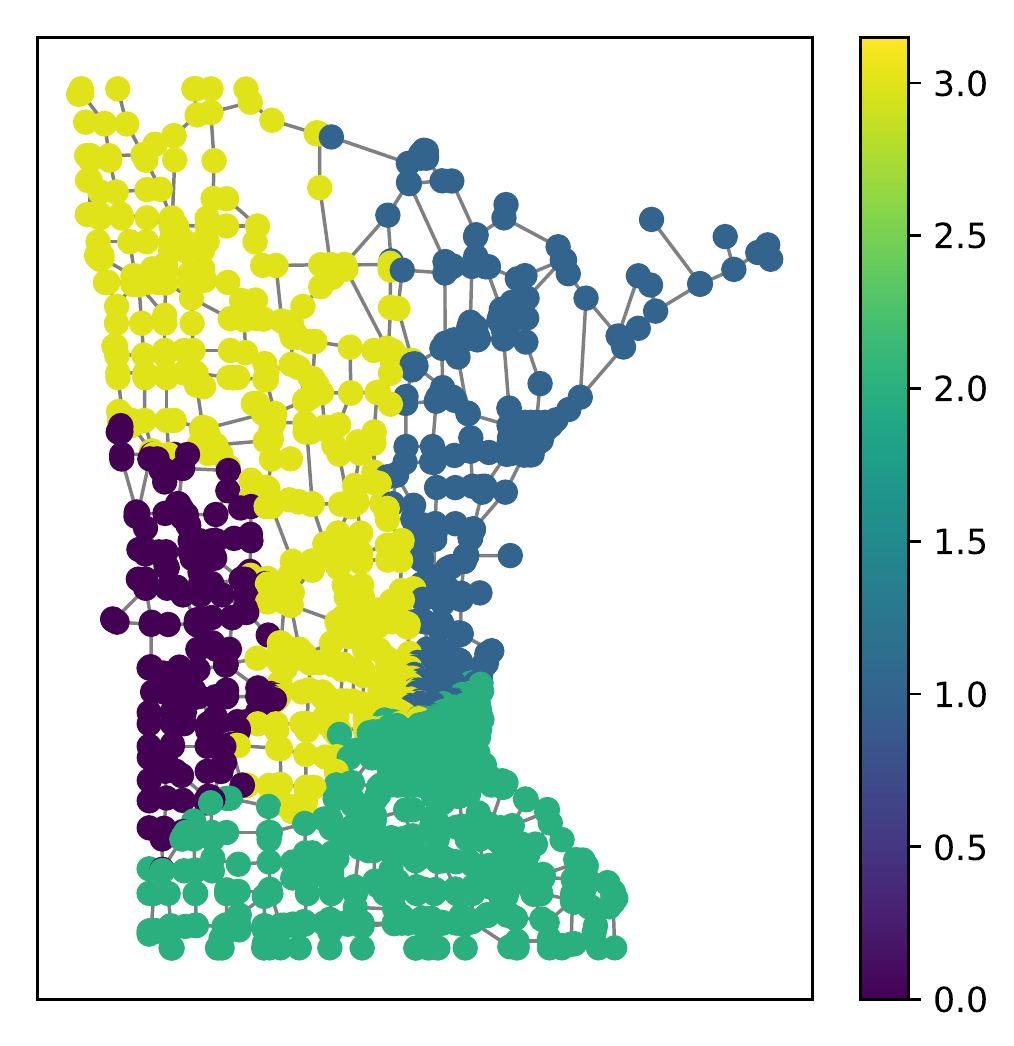} 
\end{minipage}%
\begin{minipage}[b]{0.33\linewidth}
       \includegraphics[width=\linewidth]{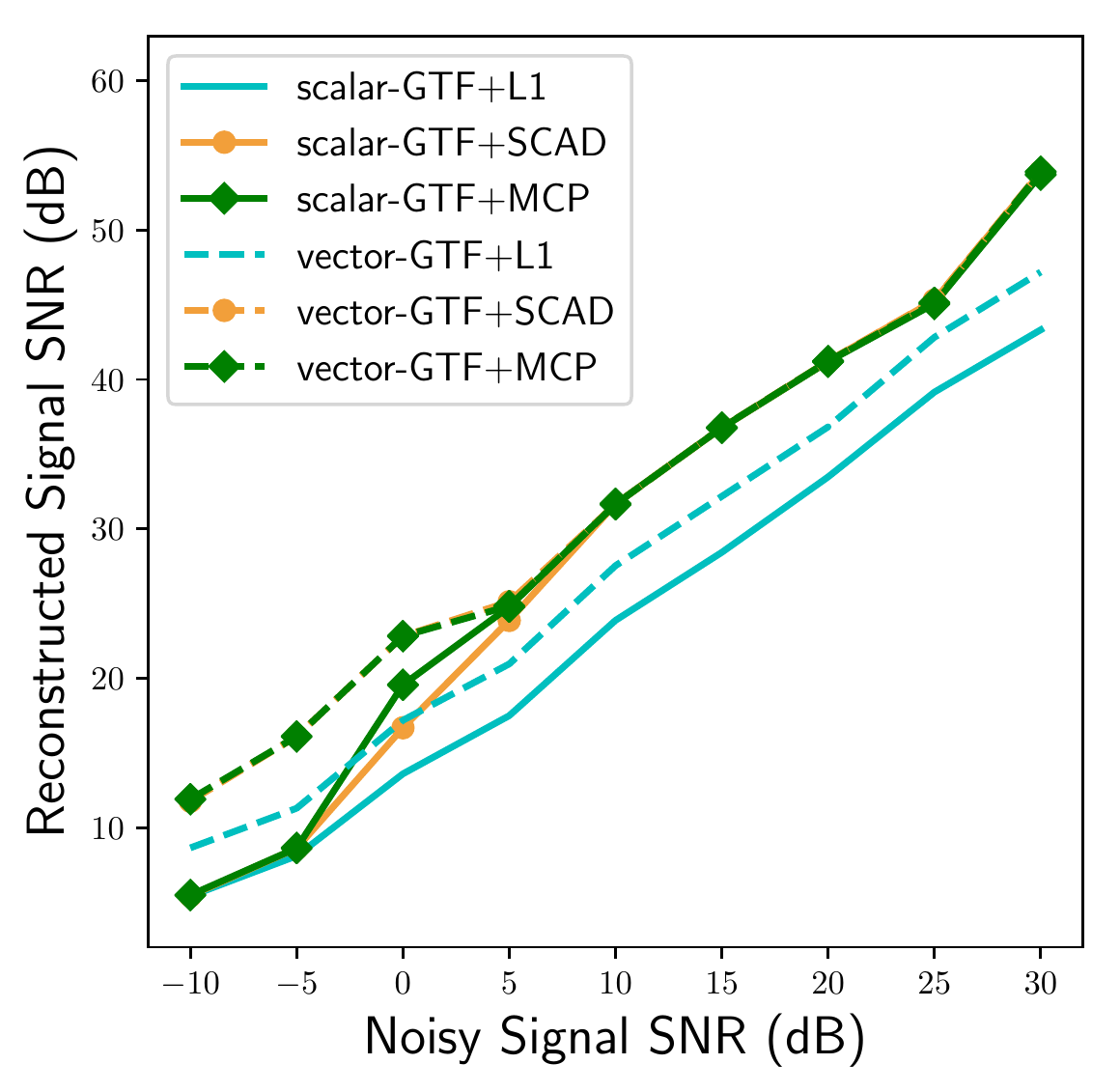}
         \includegraphics[width=\linewidth]{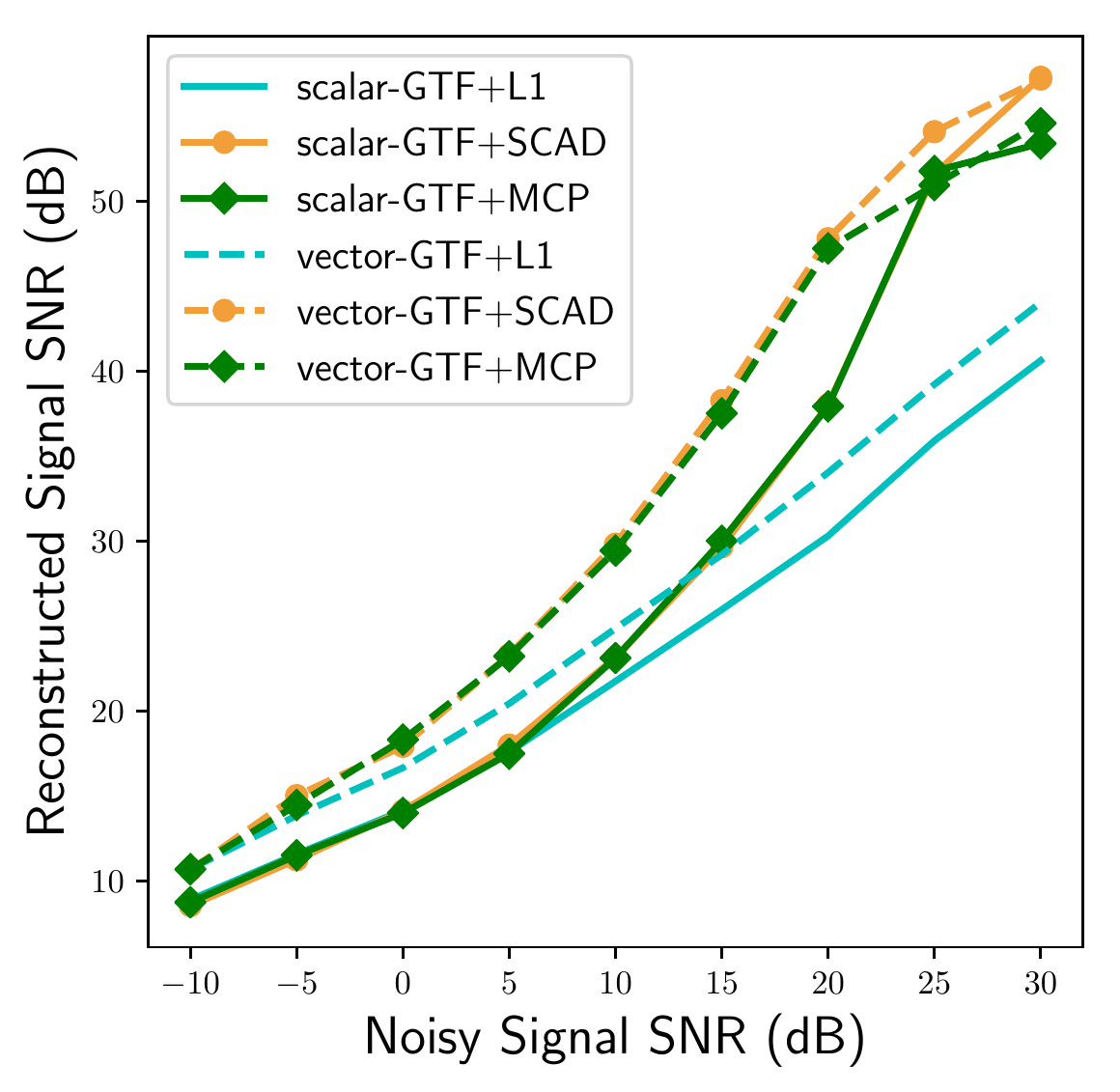}
\end{minipage}%
\begin{minipage}[b]{0.34\linewidth}
\includegraphics[width=\linewidth]{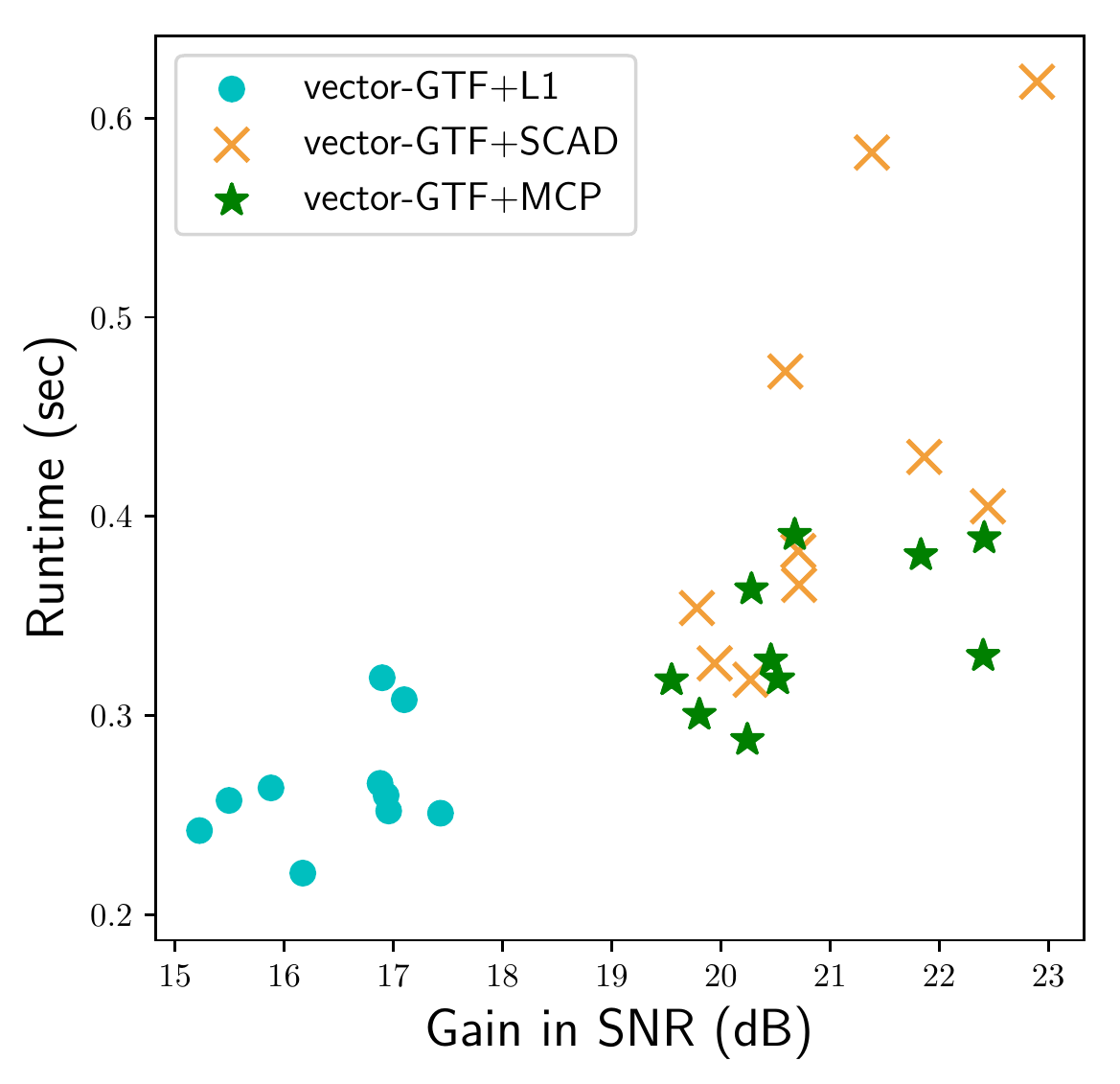}
\includegraphics[width=\linewidth]{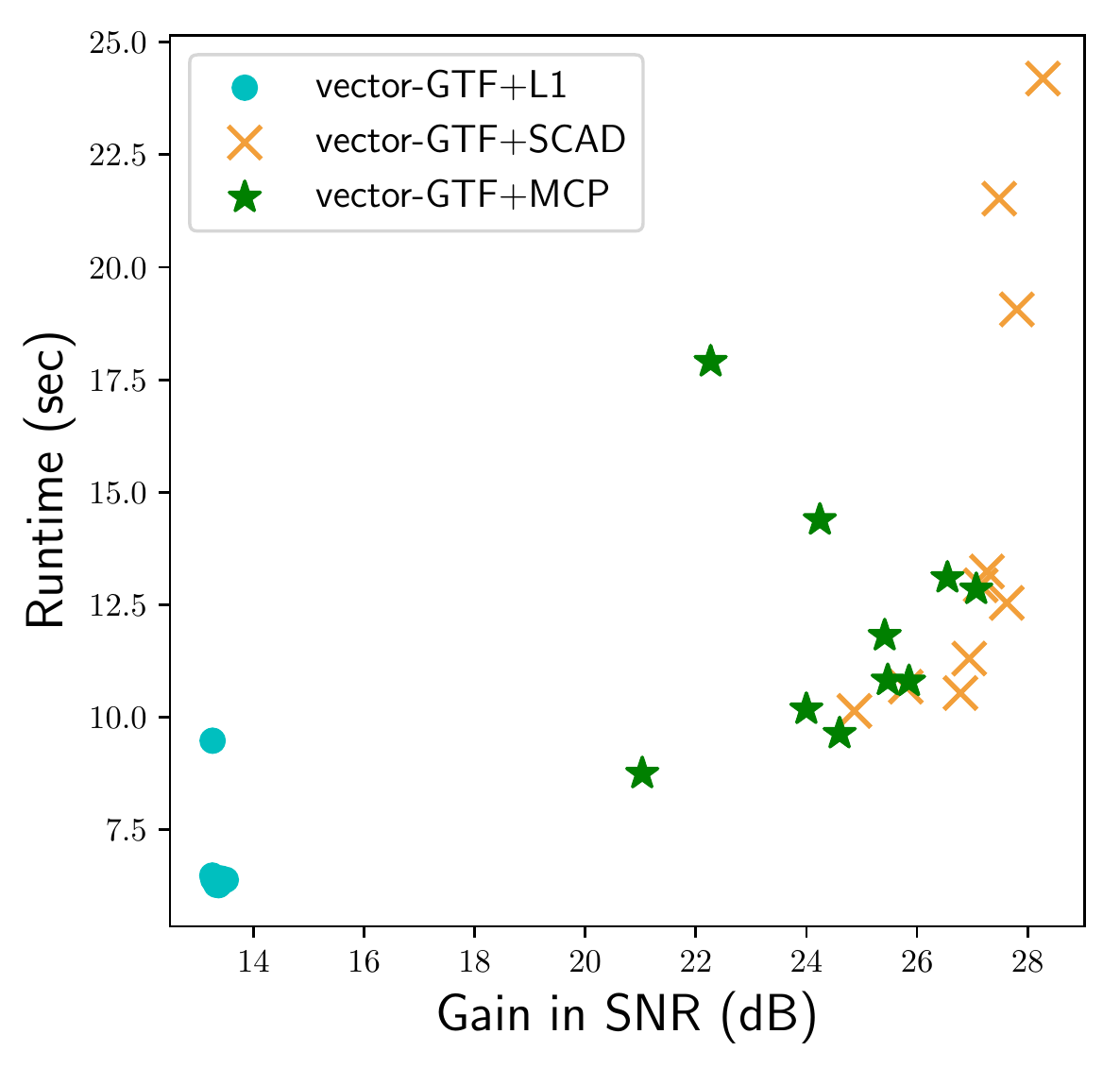}
 \end{minipage}
    \caption{The left panel shows the ground truth piecewise constant signals on $20\times 20$ 2D-grid graph (top), and Minnesota road graph (bottom). The middle panel shows their corresponding plots of input signal SNR versus reconstructed signal SNR, averaged over 10 and 20 realizations, respectively. Finally, the right panel plots the computation time against gain in SNR from denoising via vector-GTF. 10 trials were performed for each regularizer, where the input signal SNR was fixed at 20dB.
   }
   \label{fig:snr}
\end{figure*}
Then, we compare the performance of GTF using non-convex regularizers such as SCAD and MCP with that using the $\ell_1$ norm more rigorously. For the ground truth signal $\betastar$, we construct a piecewise constant signal on a $20 \times 20$ 2D-grid graph and the Minnesota road graph~\cite{defferrard_pygsp:_nodate} as shown in the left panel of Fig.~\ref{fig:snr}, and add different levels of noise following \eqref{eq:model}. We recover the signal by scalar-GTF with Alg.~\ref{alg:admmvec}, and plot the SNR of the reconstructed signal versus the SNR of the input signal \textit{in solid lines} in the middle panel of Fig.~\ref{fig:snr}, averaged over $10$ and $20$ realizations, respectively. SCAD/MCP consistently outperforms $\ell_1$ in denoising graph signals defined over both regular and irregular structures.

\subsection{Denoising Vector-valued Signals via GTF}

We compare the performance of vector-GTF in \eqref{eq:gtfvec} with \eqref{eq:sep_gtf}, which applies scalar-GTF to each column of the vector-valued graph signal. The convex $\ell_1$ norm, and the non-convex SCAD and MCP are employed. We reuse the same ground truth graph signals over the 2D-grid graph and the Minnesota road graph constructed in Section~\ref{sec:simulation_denoising} in Fig.~\ref{fig:snr}. $d$ independent noisy realizations of the graph signal are concatenated to construct a noisy vector-valued graph signal with dimension $d=10$ on the 2D-grid graph and with $d=20$ on the Minnesota road graph. We recover the vector-valued signal by minimizing vector-GTF \eqref{eq:gtfvec} with Alg.~\ref{alg:admmvec}. 

The middle panel of Fig.~\ref{fig:snr} plots the average SNR of the reconstructed signal versus the average SNR of the input signal \textit{in dotted lines}. We emphasize that the performance of \eqref{eq:sep_gtf} is the same as applying scalar-GTF to each realization, which is shown in the middle panel of Fig.~\ref{fig:snr} in solid lines. As before, SCAD/MCP consistently outperforms $\ell_1$ in denoising signals over both regular and irregular graphs. Furthermore, as expected, due to the sharing of information across realizations, vector-GTF consistently outperforms scalar-GTF, especially in the low SNR regime. 

The right panel of Fig.~\ref{fig:snr} plots the computation time versus the gain in SNR from denoising via vector-GTF. 10 trials are performed for each regularizer with the input signal SNR fixed at 20dB. Parameter tuning and eigen decomposition of $\bm{\Delta}^{(2)}$ are preprocessing steps, and hence they are not included in the time measurement; but for reference, the eigen decomposition took 0.025 and 2.5 seconds for 2D-grid and Minnesota graphs, respectively. Since GTF with non-convex regularizers are warm-started by the $\ell_1$ estimate, the runtime for $\ell_1$ GTF is added to the SCAD/MCP runtime. Overall, running vector-GTF with SCAD/MCP after once with $\ell_1$ takes more time, but with large benefits in the denoising performance. Even with the additional computation time, Vector-GTF runs reasonably fast; with the Minnesota road network, where $n=2642$ and $m=3304$, computation takes less than 25 seconds.

\begin{table}[h]
 \footnotesize
  \begin{center}
    {\setlength{\tabcolsep}{0.5em}
    \begin{tabular}{@{}rccccccccc@{}}
      \toprule \addlinespace[2mm]  
        &   \bf{Average} &   \bf{\#1}   &    \bf{\#2}  &     \bf{\#3}  &  \bf{\#4}  &    \bf{\#5}   &    \bf{\#6}  &   \bf{\#7}  &   \bf{\#8}       \\    \addlinespace[1mm]  \midrule \addlinespace[1mm]
Input SNR (dB)&8.7 & -14 & 0 &  0 &  3.5 &   5.8 &  12 &  29 &  34 \\\midrule \addlinespace[1mm] 
Vector-GTF + $\ell_1$ & 29 & \textbf{10} & \textbf{20} & \textbf{23} & \textbf{26} &  \textbf{36} & \textbf{37} & 39 & 38\\\addlinespace[1mm]
Scalar-GTF + $\ell_1$& 21 & 0 & 11 & 13 & 16 &  18 & 26 & 41 & 45 \\\addlinespace[1mm]
Vector-GTF + SCAD& \textbf{32} & \textbf{10 }& \textbf{20} & 22 & 25 & \textbf{36} & 35 & \textbf{49} & \textbf{61}  \\\addlinespace[1mm]
Scalar-GTF + SCAD& 29 & 0 & 15 & 17 & 25 & 35 & 34 & 47 & 60\\\addlinespace[1mm]
Vector-GTF + MCP& \textbf{32} & \textbf{10} & \textbf{20} & 22 &  25 & \textbf{36} & 35 & \textbf{49} & \textbf{61} \\\addlinespace[1mm]
Scalar-GTF + MCP& 29 &  0 & 15 & 22 & 24 & 30 & 33 & \textbf{49} & 60 \\
      \bottomrule
    \end{tabular}}
  \end{center}
  \caption{\label{table:mmv}
  Noisy input and reconstructed signal SNRs for eight measurements of varying input SNRs, rounded to two significant figures. Highest reconstructed signal SNR for each measurement is in \textbf{bold}.
 }
\end{table}

We further investigate the benefit of sharing information across measurements or realizations in the following experiment, using the same ground truth signal on the 2D-grid graph. We stack eight noisy realizations of this same piecewise constant signal to build a vector-valued signal. We construct these noisy measurements by scaling each one of them differently and randomly such that each will have SNR $\sim \log_{10}\mbox{Uniform}[-10,30]$dB under \eqref{eq:model}. This has the effect of rendering some measurements more~\emph{informative} than others, and potentially allowing vector-GTF to reap the benefits of sharing information across  measurements. We recover the 8-dimensional graph signal via Alg.~\ref{alg:admmvec} using $\ell_1$, SCAD, and MCP regularizers, and in Table~\ref{table:mmv}, report the input signal and reconstructed signal SNRs for each measurement in addition to the average SNRs. $\lambda$ is fixed at $0.5\sigma^2$.

First of all, notice that as before, using SCAD/MCP generally achieves results with higher SNR than using $\ell_1$, and that on average, minimizing \eqref{eq:gtfvec} outperforms minimizing \eqref{eq:sep_gtf}. The effect of sharing information across measurements is most apparent in low SNR settings, when information about the boundaries of the graph signal can be borrowed from higher SNR signals to improve the estimation. On the other hand, sharing information with noisier signals does not help denoising signals with high input SNR. However, it is worth noting that, unlike $\ell_1$, SCAD/MCP does not see decrease in its performance in the high SNR settings.

\subsection{Event Detection with NYC Taxi Data}

\begin{figure}[t]
  \begin{minipage}[b]{\linewidth}
    \centering
    \includegraphics[width=0.95\linewidth, height=11cm]{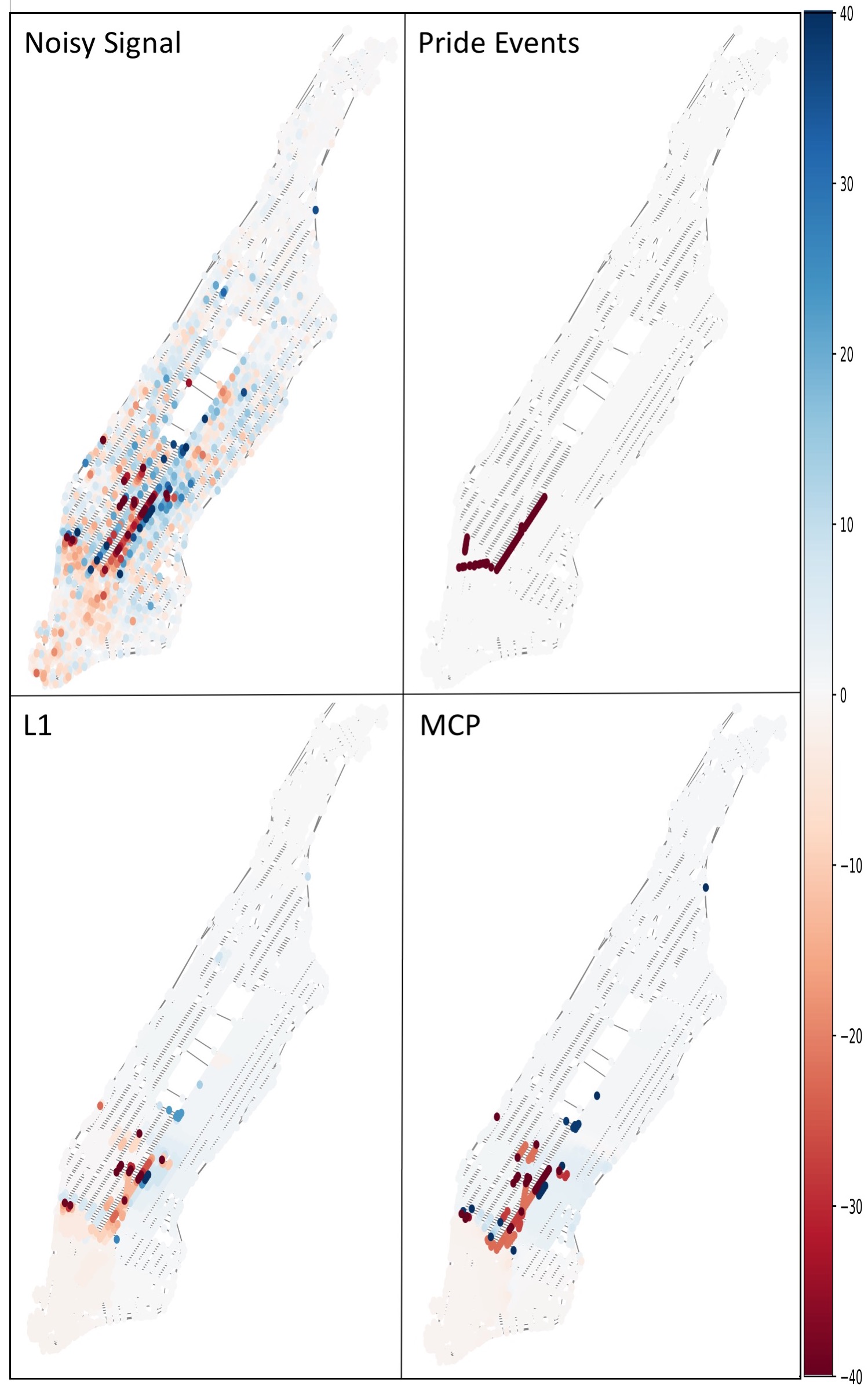}
    \caption{Top left: the noisy signal on the Manhattan road network is the change in the taxi pickup and dropoff count during the 2011 NYC Gay Pride. Top right: areas of Pride events, where the traffic was blocked off. Bottom: the GTF estimates using $\ell_1$ and MCP. The GTF estimate with MCP better detects and localizes the event, compared to the one using $\ell_1$ penalty.}
    \label{fig:nyc_final}
  \end{minipage}
\end{figure}
To further illustrate graph trend filtering on a real-world dataset, we consider the road network of Manhattan where the nodes correspond to junctions~\cite{boeing_osmnx_2017}. We map the pickups and dropoffs of the NYC taxi trip dataset~\cite{taxi_data, socrata} to the nearest road junctions, and define the total count at that junction to be the signal value on the corresponding graph node. The signal of interest, plotted on the top left panel of Fig.~\ref{fig:nyc_final}, is the difference between the \textit{event} graph signal on the day of NYC Gay Pride parade, 12-2pm on June 26, 2011, and the \textit{seasonal average} graph signal at the same time during the 8 nearest Sundays. During the event, no pickups and dropoffs could occur in the areas shown in the top right panel of Fig.~\ref{fig:nyc_final} \cite{nycpride}.
We denoise the signal via GTF using both $\ell_1$ and MCP, where we chose $\lambda$ such that $\|\bm{\Delta}\betahat\|_0 \approx 200$.
Once again, we observe the GTF estimate with MCP produces sharper traces around the parade route, indicating better capabilities of event detection and localization.

\subsection{Semi-supervised Classification}

\begin{table}[h]
\footnotesize
\begin{center}
{\setlength{\tabcolsep}{0.5em}
\begin{tabular}{@{}llrrrrrr@{}}
\toprule \addlinespace[1mm]
 \multicolumn{2}{c}{}
& \textbf{Heart}&\textbf{Wine quality} & \textbf{Wine}  & \textbf{Iris} & \textbf{Breast} &\textbf{Car}  \\ \addlinespace[1mm] \midrule 
\multicolumn{2}{c}{\# of samples ($n$)}                                                                    
& 303 & 1599 & 178  & 150  & 569 & 1728  \\ 
\multicolumn{2}{c}{ \# of classes ($K$)}                                                                  
& 2 & 6 & 3  & 3  & 2 & 4  \\ \addlinespace[1mm] \midrule
\multirow{5}{*}{$k=0$} & $\ell_1$                                                                         
& 0.148 & 0.346        & 0.038 & 0.036 & 0.042  & 0.172 \\\addlinespace[1mm] 
                   & \multirow{2}{*}{\begin{tabular}[c]{@{}l@{}}SCAD\\\addlinespace[1mm] p-value\end{tabular}} 
& 0.148 & \textit{0.353}        & 0.038 & 0.033 & 0.042  & \textbf{0.149}  \\\addlinespace[1mm] 
                   &                                                                            
& 1.  & \textit{0.06}            & 1.  & 0.27  & 1.   & \textbf{0.06}   \\\addlinespace[1mm]  
                   & \multirow{2}{*}{\begin{tabular}[c]{@{}l@{}}MCP\\ \addlinespace[1mm]p-value\end{tabular}}  
& 0.144 & 0.351        & 0.037 & 0.035 & 0.040  &  \textbf{0.148} \\\addlinespace[1mm]
                   &                                                                            
& 0.23  & 0.18         & 0.34  & 0.34  & 0.35   & \textbf{0.05}   \\\addlinespace[1mm]
\midrule
\multirow{5}{*}{$k=1$} & $\ell_1$                                                                        
& 0.143 & 0.351        & 0.034 & 0.039 & 0.035  & 0.104 \\\addlinespace[1mm]  
                   & \multirow{2}{*}{\begin{tabular}[c]{@{}l@{}}SCAD\\\addlinespace[1mm] p-value\end{tabular}} 
& 0.144 & 0.350        & 0.034 & 0.039 & 0.035  & 0.104 \\\addlinespace[1mm] 
                   &                                                                            
& 0.30  & 0.43          & 0.34  & 1.  & 0.71   & 0.66    \\\addlinespace[1mm] 
                   & \multirow{2}{*}{\begin{tabular}[c]{@{}l@{}}MCP\\\addlinespace[1mm] p-value\end{tabular}}  
& \textit{0.146} & 0.350        & 0.034 & 0.039 & \textbf{0.034}  & 0.103  \\ \addlinespace[1mm]
                   &                                                                            
& \textit{0.05}  & 0.44           & 0.34  & 1.  & \textbf{0.02}  & 0.23   \\ \bottomrule
\end{tabular}}
\end{center}
\caption{Misclassification rates averaged over 10 trials, with p-values from running sampled t-tests between SCAD/MCP misclassification rates and the corresponding rates using $\ell_1$. Cases where non-convex penalties perform better than $\ell_1$ with p-value below $0.1$ are highlighted in \textbf{bold}, and where they perform worse are in \textit{italic}.   }
\label{ssl}
\end{table}

Graph-based learning provides a flexible and attractive way to model data in semi-supervised classification problems when vast amounts of unlabeled data are available compared to labeled data, and labels are expensive to acquire~\cite{belkin_regularization_2004,zhu_semi-supervised_2003, talukdar_new_2009}. One can construct a nearest-neighbor graph based on the similarities between each pair of samples, and hope to propagate the label information from labeled samples to unlabeled ones. We move beyond our original problem in \eqref{eq:gtf} to a $K$-class classification problem in a semi-supervised learning setting, where for a given dataset with $n$ samples, we observe a subset of the one-hot encoded class labels, $\bm{Y} \in \R^{n\times K}$, such that $Y_{ij} = 1$ if $i$th sample has been observed to be in $j$th class, and $Y_{ij} = 0$ otherwise. A diagonal indicator matrix $\bm{M} \in \R^{n \times n}$ denotes samples whose class labels have been observed.  Then, we can define the modified absorption problem~\cite{zhu_semi-supervised_2003,wang_trend_2016,talukdar_new_2009} using a variation of GTF to estimate the unknown class probabilities $\B\in\mathbb{R}^{n\times K}$:
\begin{align}
\widetilde{\B}   =& \argmin_{\B\in \R^{n\times K}} \frac{1}{2}\|\bm{M}(\bm{Y}-\B)\|_{\mathrm{F}}^2  \nonumber \\
    &+\sum_{j=1}^{K} g(\Dk \B_{\cdot j}; \lambda, \gamma)+  \epsilon \|\bm{R}-\B\|_{\mathrm{F}}^2 ,
    \label{eq:ssl}
\end{align}
where $\bm{R} \in \R^{n \times K}$ (set to be uniform in the experiment) is a fixed prior belief, and $\epsilon >0$ determines how much emphasis to be given to the prior belief. The labels $\widetilde{\bm{Y}}$ can be estimated using $\widetilde{\B}$ such that $\widetilde{Y}_{ij}=1$ if and only if $j = \arg\max_{1\leq \ell \leq K}\widetilde{B}_{i\ell}$, and otherwise $\widetilde{Y}_{ij}=0$. Note that this can be completely separated into $K$ scalar-GTF problems, one corresponding to each class. 

We applied the algorithm in $\eqref{eq:ssl}$ to 6 popular UCI classification datasets~\cite{asuncion_uci_2007} with $\epsilon = 0.01$. For each dataset, we normalized each feature to have zero mean and unit variance, and constructed a 5-nearest-neighbor graph of the samples based on the Euclidean distance between their features, with edge weights from the Gaussian radial basis kernel. We observed the labels of 20\% of samples in each class randomly. Table \ref{ssl} shows the misclassification rates averaged over $10$ repetitions, which demonstrates that the performance using non-convex penalties such as SCAD/MCP are at least competitive with, and often better than, those with the $\ell_1$ penalty.

\section{Conclusions} \label{sec:conclusions}
We presented a framework for denoising piecewise smooth signals on graphs that generalizes the graph trend filtering framework to handle vector-valued signals using a family of non-convex regularizers. We provided theoretical guarantees on the error rates of our framework, and derived a general ADMM-based algorithm to solve this generalized graph trend filtering problem. Furthermore, we demonstrated the superior performance of these non-convex regularizers in terms of reconstruction error, bias reduction, and support recovery on both synthetic and real-world data. In particular, its performance on semi-supervised classification is investigated. In future work, we plan to further study this approach when the graph signals are observed by indirect and incomplete measurements.

\bibliographystyle{IEEEtran}
\label{sec:refs}
\bibliography{ms.bib}

\appendix

\subsection{Proof of Theorem \ref{thm:oracle}}\label{proof_oracle}
\begin{proof}
We denote $\D =\Dk$. Define $R$ as the row space of $\D$, and $R^{\bot}$ the null space. Let $\mathcal{P}_\mathrm{R}=\Ddag\D$, the projection onto $R$, and $\|\x\|_\mathrm{R}=\|\mathcal{P}_\mathrm{R}\x\|_2$. Additionally, $\mathcal{P}_{\mathrm{R}^{\bot}} = \bm{I} - \Ddag\D$, the projection onto $R^{\bot}$.  
Since $\betahat$ is a stationary point of $f(\betab)$, it follows that 
\begin{align}\label{eq:kkt}
 \bm{0} \in \nabla_\betab f(\betab)|_{\betab=\betahat} = (\betahat-\y) + \nabla_\betab \glam { \D \betab}|_{\betab=\betahat} .
\end{align}
By the chain rule, $\nabla_\betab \glam { \D \betab}|_{\betab=\betahat}=\{\D^\top\bm{z}: \bm{z} \in \nabla_{\bm{x}} \glam{\bm{x}}|_{\bm{x}=\D \betahat} \}$. 
Then by \eqref{eq:kkt}, there exists $\bm{z} \in \nabla_{\bm{x}} \glam{\bm{x}}|_{\bm{x}=\D \betahat}$, such that
$$ \bm{0} = (\betahat-\y) + \D^\top\bm{z} .$$
In particular, $\forall \betab \in \R^n$, we have 
\begin{align}
\betab^\top (\y - \betahat) &= (\D \betab)^\top \bm{z}, \quad  \label{eq:z1}
\end{align}
and, specializing to $\betahat$,
\begin{align}
 \betahat^\top(\y - \betahat) & = (\D \betahat)^\top \bm{z}.  \label{eq:z2}
\end{align}
Subtract \eqref{eq:z2} from \eqref{eq:z1}, and use the definition of subgradient to get $\forall \betab \in \R^n$,
\begin{align} \label{eq:subgradient}
\betab^\top( \y - \betahat )  - \betahat^\top( \y - \betahat ) &  =( \D \betab - \D \betahat)^\top \bm{z}  \nonumber \\ 
&\leq   \glam{ \D \betab} - \glam { \D \betahat}  .
\end{align}
By the measurement model $\y = \betastar + \bm{\epsilon}$ and the polarization equality, i.e. $ 2\bm{a}^\top \bm{b} = \|\bm{a}\|_2^2 + \| \bm{b} \|_2^2 - \| \bm{a} - \bm{b} \|_2^2 $, the left-hand side of \eqref{eq:subgradient} can be rewritten as
\begin{align}
 & \quad \betab^\top( \y - \betahat )  - \betahat^\top( \y - \betahat ) \nonumber \\
 & = (\betab - \betahat)^\top( \betastar - \betahat ) + \bm{\epsilon}^\top(\betab - \betahat) \nonumber \\
	& = \frac{1}{2} \| \betab - \betahat\|_2^2 + \frac{1}{2}  \| \betastar - \betahat \|_2^2 -\frac{1}{2}  \| \betab -\betastar \|_2^2 + \bm{\epsilon}^\top(\betab - \betahat). \label{eq:expand}
\end{align}
Combining \eqref{eq:subgradient} and \eqref{eq:expand} gives us $\forall \betab \in \R^n$
\begin{align}\label{eq:bigineq}
&\quad \| \betahat- \betab\|_2^2 +  \| \betahat - \betastar\|_2^2 \nonumber \\
&\leq  \| \betab -\betastar \|_2^2 + 2 \bm{\epsilon}^\top(\betahat -\betab)  + 2\glam { \D \betab} - 2 \glam {\D \betahat} . 
\end{align}
Let us first consider $\bm{\epsilon}^\top (\betahat -\betab) $. From H\"older's inequality,
\begin{align}
	&	\bm{\epsilon}^\top (\betahat - \betab)
 =(\D^{\dagger}\D\bm{ \epsilon})^\top(\betahat - \betab) +(\mathcal{P}_{\mathrm{R}^{\bot}}\bm{ \epsilon})^\top(\betahat - \betab) \nonumber\\
&	\leq \| (\D^{\dagger})^\top \bm{ \epsilon} \|_\infty \| \D ( \betahat - \betab) \|_1 +\|\mathcal{P}_{\mathrm{R}^{\bot}}\bm{ \epsilon} \|_2 \| \betahat - \betab \|_2.  \label{eq:err}
\end{align}
By standard tail bounds for independent Gaussian random variables, we have with probability at least $1-\delta$,
\begin{equation}\label{eq:tail}
\|(\Ddag)^\top\bm{\epsilon}\|_\infty \le \sigma \zeta_k \sqrt{2 \log (\frac{er}{\delta})}.
\end{equation}
Additionally, recognize that $\|\bm{\epsilon}\|_{\mathrm{R}^{\bot}}^2$ is a chi-squared random variable with $C_G$ degrees of freedom. We can then invoke the one-sided tail bound for chi-squared random variables (c.f. \cite[Example 2.5]{wainwright2019high}) such that for any $0 \le t \le 1$,
\begin{align*}
    P\big( \|\bm{\epsilon}\|_{\mathrm{R}^{\bot}}^2 \geq \sigma^2 C_G (1 + t) \big) \leq \exp \left(\frac{-C_G t^2}{8}\right).
\end{align*}
Consequently, with probability at least $1-\delta$,
\begin{equation}
\label{eq:bound_chi}
  \|\bm{\epsilon}\|_{\mathrm{R}^{\bot}}^2 \leq  \sigma^2 \Big( C_G + 2 \sqrt{2 C_G \log(1/\delta)} \Big) .
\end{equation}
The inequalities \eqref{eq:bound_chi} and \eqref{eq:tail} hold simultaneously with probability at least $1-2\delta$. 
Then, using $\lambda \|\bm{x}\|_1 \le g(\bm{x}) + \frac{\mu}{2}\|\bm{x}\|_2^2$ and $\lambda = \sigma \zeta_k \sqrt{2 \log (\frac{er}{\delta})}\geq \| (\D^{\dagger})^\top \bm{ \epsilon} \|_\infty$, we can bound \eqref{eq:err} further as
\begin{align}
    &\bm{\epsilon}^\top (\betahat - \betab)  \le   \| \mathcal{P}_{\mathrm{R}^{\bot}} \bm{\epsilon} \|_2  \| \betahat - \betab \|_2 + \lambda\| \D ( \betahat - \betab) \|_1 \nonumber \\ 
  &  \le  \| \mathcal{P}_{\mathrm{R}^{\bot}} \bm{\epsilon} \|_2 \| \betahat - \betab \|_2 + \glam{\D(\betahat - \betab)} +  \frac{\mu}{2}\|\D(\betahat - \betab)\|_2^2. \nonumber 
\end{align}
Together with $ \|\D(\betahat - \betab)\|_2^2 \le \|\D\|^2\|(\betahat - \betab)\|_2^2$, we can upper bound \eqref{eq:bigineq} as
\begin{align}
& \|\betahat- \betab\|_2^2 +  \| \betahat  - \betastar \|_2^2 \nonumber \\
& \leq  \| \betab -\betastar \|_2^2 + 2\| \mathcal{P}_{\mathrm{R}^{\bot}} \bm{\epsilon} \|_2 \| \betahat - \betab \|_2 + \mu \| \D  \|^2 \| \betahat - \betab \|_2^2 \nonumber   \\
 & \qquad + 2 \glam { \D ( \betahat - \betab) } + 2 \glam {\D \betab} - 2 \glam { \D \betahat}. \label{eq:err2}
\end{align}
Note that for any set $T$, $\glam{\x} = \sum_{i \in T} \rlam{x_i}+\sum_{j\in T^c}\rlam{x_j} = \glam{(\x)_T}+ \glam{(\x)_{T^c}} $. Therefore, using the triangle inequality and subadditivity and symmetry of $\rho$, 
\begin{align}
	  & \glam { \D ( \betahat - \betab) }  +  \glam { \D \betab}-  \glam { \D \betahat}\nonumber  \\ 
	  & \le \glam { (\D ( \betahat - \betab))_T} + \glam { (\D \betab)_{T^c} } + \glam { (\D \betahat)_{T^c} }\nonumber \\
	  &\quad+ \glam {\D\betab} -  \glam { (\D \betahat)_T}- \glam {(\D \betahat)_{T^c}} \nonumber \\
	  & = \glam { (\D ( \betahat - \betab))_T } + 2\glam {(\D \betab)_{T^c}}  + \glam {( \D \betab)_T} -  \glam { (\D \betahat)_T} \nonumber \\ 
	  & \leq 2 \glam{ (\D( \betahat - \betab))_T}  + 2 \glam { (\D \betab)_{T^c}} . \label{eq:triangle}
\end{align}
We bound \eqref{eq:triangle} further by the compatibility factor,
\begin{align}
\glam{ (\D( \betahat - \betab))_T} & \leq  \lambda \|(\D( \betahat - \betab))_T \|_1  \nonumber \\
&\leq \lambda \sqrt{|T|} \kappa_{T}^{-1} \| \betahat - \betab \|_2.  \label{eq:kappa}
\end{align}
Now combining \eqref{eq:err2}, \eqref{eq:triangle}, and \eqref{eq:kappa}, we then have
\begin{align}
\| \betahat-\betab\|_2^2 +&  \| \betahat - \betastar\|_2^2 \leq \| \betab -\betastar \|_2^2 + 4 \glam {(\D \betab)_{T^c}} \nonumber\\
&+ 2\left( \|\mathcal{P}_{\mathrm{R}^{\bot}}\bm{ \epsilon} \|_2  +  2 \lambda \sqrt{|T|} \kappa_{T}^{-1} \right ) \| \betahat - \betab \|_2 \nonumber \\
& + \mu \| \D  \|^2 \| \betahat - \betab \|_2^2 .\nonumber
\end{align}
Apply Young's inequality, which is $ 2ab \leq  a^2/\epsilon + \epsilon b^2 $ for $\epsilon > 0$, with $a=  \|\mathcal{P}_{\mathrm{R}^{\bot}}\bm{ \epsilon} \|_2  +  2 \lambda \sqrt{|T|} \kappa_{T}^{-1}$, $b = \| \betahat - \betab \|_2$, and $\epsilon=1- \mu \| \D  \|^2 >0$, we have
\begin{align} 
&2\left( \|\mathcal{P}_{\mathrm{R}^{\bot}}\bm{ \epsilon} \|_2  +  2 \lambda \sqrt{|T|} \kappa_{T}^{-1} \right ) \| \betahat - \betab \|_2 \nonumber\\
&\le \frac{1}{\epsilon} \left( \|\mathcal{P}_{\mathrm{R}^{\bot}}\bm{ \epsilon} \|_2  +  2 \lambda \sqrt{|T|} \kappa_{T}^{-1} \right )^2 + \epsilon \| \betahat - \betab \|_2^2 \nonumber
\\
&\le \frac{2}{(1- \mu \| \D  \|^2)} \left(\|\mathcal{P}_{\mathrm{R}^{\bot}}\bm{ \epsilon} \|_2^2  +  4 \lambda^2 |T| \kappa_{T}^{-2} \right)\label{eq:epsilon} \\
&\quad+ (1- \mu \| \D  \|^2) \| \betahat - \betab \|_2^2. \nonumber
\end{align}
Therefore,
\begin{align}
&\| \betahat- \betab\|^2 +  \| \betahat -\betastar\|_2^2 \nonumber\\
& \leq   \| \betab -\betastar \|_2^2 + 4 \glam {(\D \betab)_{T^c}}+  \| \betahat- \betab\|_2^2\nonumber\\
& +  \frac{2}{(1- \mu \| \D  \|^2)} \left(\|\mathcal{P}_{\mathrm{R}^{\bot}}\bm{ \epsilon} \|_2^2  +  4 \lambda^2 |T| \kappa_{T}^{-2} \right).
\end{align}
Cancel $ \| \betahat- \betab\|_2^2$ on both sides, apply the infimum over $\betab$ and plug in the bounds \eqref{eq:bound_chi} to get
\begin{align*}
	 & \| \betahat - \betastar \|_2^2 \leq \inf_{\betab} \left\{ \| \betab -\betastar \|_2^2 +4\glam { (\D \betab)_{T^c}} \right\} \nonumber \\ 
	 &+ \frac{2\sigma^2}{ (1 -  \mu  \| \D \|^2)} \left(   C_G + 2 \sqrt{2 C_G \log(\frac{1}{\delta})}  + \frac{8 \zeta_k^2 |T|}{\kappa_{T}^2} \log (\frac{er}{\delta}) \right ).
\end{align*}

The proof extends for the vector-GTF \eqref{eq:gtfvec} in a similar manner. We need to replace \eqref{eq:err} by
\begin{align}
	& \langle	\bm{E}, \Bhat - \B \rangle
 =\langle \D^{\dagger}\D\bm{E}, \Bhat - \B\rangle +\langle \mathcal{P}_{\mathrm{R}^{\bot}}\bm{E}, \Bhat - \B \rangle \nonumber\\
&\le \lambda \sum_{\ell=1}^r  \big\|\D_{\ell\cdot}(\Bhat - \B)\big\|_2 +\|\mathcal{P}_{\mathrm{R}^{\bot}}\bm{E} \|_\mathrm{F} \| \Bhat - \B \|_\mathrm{F}, \nonumber
\end{align}
where $\|\mathcal{P}_{\mathrm{R}^{\bot}}\bm{E} \|_\mathrm{F}^2 \leq  d \sigma^2 \Big( C_G + 2 \sqrt{2 C_G \log(d/\delta)} \Big) $ with probability at least $1-\delta$. Similarly, for \eqref{eq:kappa}, we use the generalized definition of the compatibility factor $\kappa_T$, given as
\begin{align*}
\hlam{ (\D( \Bhat - \B))_T} & \leq  \lambda \sum_{\ell\in T}\|(\D( \Bhat - \B))_{\ell\cdot} \|_2  \\
&\leq \lambda \sqrt{|T|} \kappa_{T}^{-1} \| \Bhat - \B \|_\mathrm{F} ,
\end{align*}
which will lead to the claimed bound in the theorem.
\end{proof}

\subsection{Proof of Proposition \ref{prop:kappa}}\label{proof_prop_kappa}

\begin{proof}
By Cauchy-Schwartz inequality, we have
$$ \sum_{\ell\in T}\|(\Dk\B)_{\ell \cdot}\|_2 \le \sqrt{|T|}\|(\Dk\B)_T\|_{\mathrm{F}}, $$ 
and note that given two matrices $\bm{U}$ and $\bm{V}$, $(\bm{U}\bm{V})_T = (\bm{U})_T\bm{V}$ where $T$ is a subset of rows indices. We also use the fact that $\|\bm{U} \bm{V}\|_\mathrm{F} \leq \|\bm{U}\| \|\bm{V}\|_\mathrm{F}$. 
 We consider two cases:
\begin{itemize}[leftmargin=*]
\item For even $k$, we have
\begin{align}
&\|(\Dk\B)_T\|_\mathrm{F}= \|(\bm{\Delta})_T\bm{\Delta}^{(k)}\B\|_\mathrm{F} \nonumber\\
&\le \|(\bm{\Delta})_T\|\|\bm{\Delta}^{(k)}\B\|_\mathrm{F}  = \sqrt{\lambda_{\max}((\bm{\Delta})_T^\top(\bm{\Delta})_T))}\|\bm{\Delta}^{(k)}\B\|_\mathrm{F}. \nonumber
\end{align}
Note that $(\bm{\Delta})_T$ is equivalent to the incidence matrix of a subgraph with only $T$ edges, which allows us to bound,
\begin{align}
 \lambda_{\max}((\bm{\Delta})_T^\top(\bm{\Delta})_T))& \le  {\max_{(u,v)\in T} \{d_u + d_v\}}  \le  {2 d_{\max}}  \nonumber
\end{align}
where   $d_i$ is the degree of node $i$.
\item For odd $k$, we have
\begin{align}
&\|(\Dk\B)_T\|_\mathrm{F} = \|(\bm{\Delta}^\top)_T\bm{\Delta}^{(k)}\B\|_\mathrm{F}  \nonumber \\
&\le \|(\bm{\Delta}^\top)_T\|\|\bm{\Delta}^{(k)}\B\|_\mathrm{F} = \sqrt{\lambda_{\max}(\bm{\Delta}^{(2)}_{T \times T})}\|\bm{\Delta}^{(k)}\B \|_\mathrm{F} ,\nonumber
\end{align}
where $\bm{\Delta}^{(2)}_{T\times T}\in \R^{|T|\times |T|} $ is the principal submatrix of $\bm{\Delta}^{(2)}$ indexed by $T$. By Cauchy's interlacing theorem, the maximum eigenvalue of the submatrix is upper bounded, so
\begin{align}
& {\lambda_{\max}(\bm{\Delta}^{(2)}_{T \times T})}\le  {\lambda_{\max}(\bm{\Delta}^{(2)})} \le  {2 d_{\max}} .\nonumber
\end{align}
\end{itemize}
Therefore, for all $k$, $\|(\Dk\B)_T\|_\mathrm{F}\le \sqrt{2 d_{\max}} \|\bm{\Delta}^{(k)}\B \|_\mathrm{F}$. To conclude the proof,  note
\begin{align*}
&\sum_{\ell\in T}\|(\Dk\B)_{\ell \cdot}\|_2 \le \sqrt{|T|} \sqrt{2d_{\max}} \|\bm{\Delta}^{(k)}\B \|_\mathrm{F} \\
&\le \sqrt{|T|} \sqrt{2d_{\max}} \|\bm{\Delta}^{(k)} \| \| \B \|_\mathrm{F} \le (2d_{\max})^{\frac{k+1}{2}}  \sqrt{|T|} \|\B\|_\mathrm{F} .
\end{align*}
\end{proof}

\subsection{Proof of Theorem \ref{thm:conv}} \label{proof:convergence_admm}
We show the convergence of Alg.~\ref{alg:admmvec} by proving a modified version of \cite[Proposition 1]{ma_concave_2017}. The superscript $(m)$ denotes the values of $\B, \Z, \U$ at the $m$th iteration of the loop inside Alg.~\ref{alg:admmvec}.
\begin{myProposition}[Convergence to a feasible solution]
If $\tau \ge \mu$, then the primal residual $r^{(m)} = \|\Dk \B^{(m)} -\Z^{(m)} \|_{\mathrm{F}}$ and the dual residual $s^{(m+1)} = \|\tau (\Dk)^\top ( \Z^{(m+1)} - \Z^{(m)}) \|_{\mathrm{F}}$ of Alg.~\ref{alg:admmvec} satisfy that $\lim_{m \to \infty} r^{(m)} = 0$ and $\lim_{m \to \infty} s^{(m)} = 0$.
\end{myProposition}

\begin{proof}
Denote $\bm{D} = \Dk$, and $\rho_\lambda(\cdot) = \rho(\cdot;\lambda, \gamma)$.
Recall from Assumption~\ref{assumptions_1} (c) that there exists $\mu >0$ such that $\rho_\lambda(\|\x\|_2) + \frac{\mu}{2}\|\x\|_2^2$ is convex. Now consider the Lagrangian $\mathcal{L}(\B, \Z, \U)$ with regard to the $\ell$-th row $\z_\ell$ of $\Z =[\z_1^\top, ...\z_r^\top]^\top$, assuming all other variables are fixed:
\begin{align}
&\rho_\lambda(\|\z_\ell\|_2) + \frac{\tau}{2}\|\z_\ell-\bm{c}_1\|_2^2+c_2\nonumber \\ =&\rho_\lambda(\|\z_\ell\|_2) + \frac{\tau}{2}\|\z_\ell\|_2^2-\tau \langle \z_\ell,\bm{c}_1 \rangle +\frac{\tau}{2}\|\bm{c}_1\|_2^2+c_2 \nonumber
\end{align}
where $\bm{c}_1$ and $c_2$ represent terms of $\mathcal{L}(\B, \Z, \U)$ that do not depend on $\z_\ell$. With our choice of $\tau \ge \mu$, then $\mathcal{L}(\B, \Z, \U)$ is convex with regard to each of $\B$, $\U$, and for each row of $\Z$, allowing us to apply \cite[Theorem 5.1]{tseng_convergence_2001}. Therefore, Alg.~\ref{alg:admmvec} converges to limit points $\B^\star, \Z^{\star}, \U^{\star}$.  

Then it follows that the dual residual $\lim_{m \to \infty} s^{(m)} =\|\tau \D^\top ( \Z^{\star} - \Z^{\star})\|_{\mathrm{F}}=0$. For the primal residual, notice that the $\U$ update step in line 10 of Alg.~\ref{alg:admmvec} also shows that $\forall m, t \ge 0$,
\begin{equation*}
    \U^{(m+t)} = \U^{(m)} + \sum_{i=1}^t (\D \B^{(m+i)} -\Z^{(m+i)}) .
\end{equation*}
Fixing $t$ and setting $m \to \infty$, we have
$$ \U^{\star} = \U^{\star} + t(\D \B^{\star} -\Z^{\star}) $$
holds $\forall t\ge 0$. Hence, $\D \B^{\star} -\Z^{\star} = \bm{0}$, and therefore $\lim_{m\to \infty}r^{(m)} = \|\D \B^{\star} -\Z^{\star} \|_{\mathrm{F}}=0$. 
\end{proof}
This proposition shows that the algorithm in the limit achieves primal and dual feasibility, and that the Augmented Lagrangian in \eqref{eq:nc-gtf-lag-vec} with $\Z^{\star}$ and $\U^\star$ becomes the original GTF formulation in \eqref{eq:gtfvec}. $\B$ that is produced by every iteration of Alg.~\ref{alg:admmvec} is a stationary point of \eqref{eq:nc-gtf-lag-vec} with fixed $\Z$ and $\U$. As a result, $\B^\star$ is a stationary point of \eqref{eq:gtfvec}.

\begin{IEEEbiography}[{\includegraphics[width=1in,height=1.25in,clip,keepaspectratio]{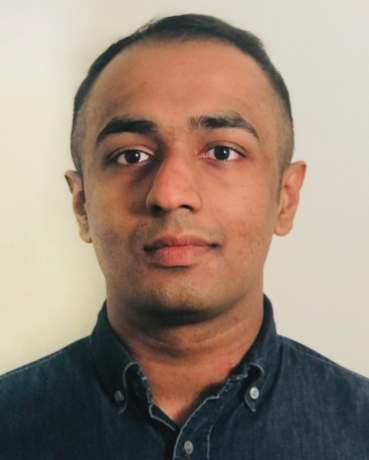}}]{Rohan Varma} is a Ph.D student in Electrical and Computer Engineering from Carnegie Mellon University and is advised by Jelena Kova\v{c}evi\'{c}. He received a  M.Sc in Electrical and Computer Engineering in 2016 from Carnegie Mellon University and a B.Sc in Electrical Engineering and Computer Science, B.A Economics and B.A Statistics in 2014 from U.C Berkeley. Varma was the recipient of the 2019 IEEE Signal Processing Society Young Author Best Paper Award. His research interests include graph signal processing, graph neural networks and more generally using graphs to augment machine learning tasks.
\end{IEEEbiography}

\begin{IEEEbiography}
[{\includegraphics[width=1in,height=1.25in,clip,keepaspectratio]{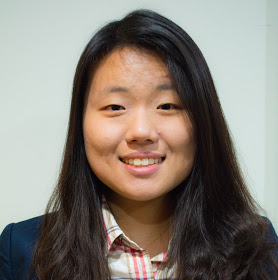}}]{Harlin Lee} is a Ph.D. candidate in Electrical and Computer Engineering at Carnegie Mellon University. She received her M.Eng. (with concentration in AI) and B.S. degrees in 2017 and 2016 respectively, both in Electrical Engineering and Computer Science from MIT. Her research interests include graph regularization, unsupervised learning, and clinical data analysis.
\end{IEEEbiography}

 \begin{IEEEbiography}
[{\includegraphics[width=1in,height=1.25in,clip,keepaspectratio]{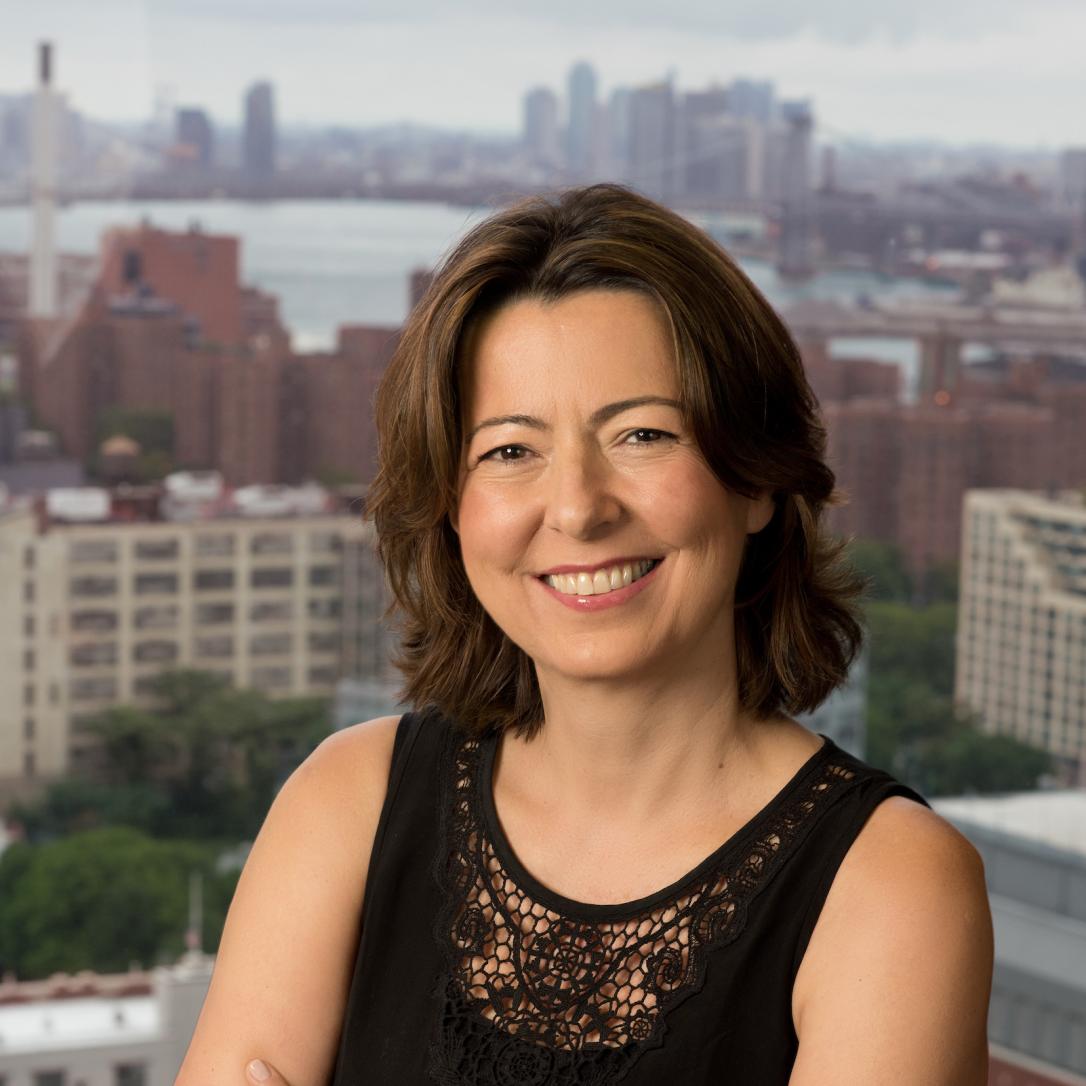}}]{Jelena Kova\v{c}evi\'{c}}
(S'88--M'91--SM'96--F'02) received the Dipl. Electr. Eng. degree from the EE Department, University of Belgrade, Yugoslavia, in 1986, and the M.S. and Ph.D. degrees from Columbia University, New York, in 1988 and 1991, respectively. From 1991--2002, she was with Bell Labs, Murray Hill, NJ. She was a co-founder and Technical VP of xWaveforms, based in New York City and an Adjunct Professor at Columbia University. In 2003, she joined Carnegie Mellon University, where she was Hamerschlag University Professor and Head of Electrical and Computer Engineering, Professor of Biomedical Engineering, and was the Director of the Center for Bioimage Informatics at Carnegie Mellon University. Since 2018, she is the William R. Berkley Professor and Dean of the Tandon School of Engineering at New York Univeristy, New York City. Her research interests include wavelets, frames, graphs, and applications to bioimaging and smart infrastructure. 

Dr.~Kova\v{c}evi\'{c} coauthored the books Wavelets and Subband Coding (Prentice Hall, 1995) and Foundations of Signal Processing (Cambridge University Press, 2014), a top-10 cited paper in the Journal of Applied and Computational Harmonic Analysis, and the paper for which A.~Mojsilovi\'{c} received the Young Author Best Paper Award. Her paper on multidimensional filter banks and wavelets was selected as one of the Fundamental Papers in Wavelet Theory. She received the Belgrade October Prize in 1986, the E.I. Jury Award at Columbia University in 1991, and the 2010 CIT Philip L. Dowd Fellowship Award from the College of Engineering at Carnegie Mellon University and the 2016 IEEE SPS Technical Achievement Award. She is a past Editor-in-Chief of the IEEE Transactions on Image Processing, served as a guest co-editor on a number of special issues and is/was on the editorial boards of several journals.  She was a regular member of the NIH Microscopic Imaging Study Section and served as a Member-at-Large of the IEEE Signal Processing Society Board of Governors.  She is a past Chair of the IEEE Signal Processing Society Bio Imaging and Signal Processing Technical Committee. 
\end{IEEEbiography}

\begin{IEEEbiography}
	[{\includegraphics[width=1in,height=1.25in,clip,keepaspectratio]{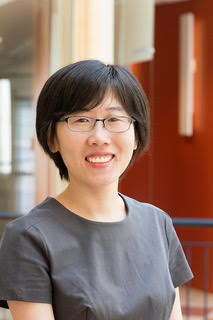}}]{Yuejie Chi}
(S'09-M'12-SM'17)
received the Ph.D. degree in Electrical Engineering from Princeton University in 2012, and the B.E. (Hon.) degree in Electrical Engineering from Tsinghua University, Beijing, China, in 2007. She was with The Ohio State University from 2012 to 2017. Since 2018, she is an Associate Professor with the department of Electrical and Computer Engineering at Carnegie Mellon University, where she holds the Robert E. Doherty Early Career Professorship. Her research interests include signal processing, statistical inference, machine learning, large-scale optimization, and their applications in data science, inverse problems, imaging, and sensing systems.

She is a recipient of the PECASE Award, NSF CAREER Award, AFOSR and ONR Young Investigator Program Awards, Ralph E. Powe Junior Faculty Enhancement Award, Google Faculty Research Award, IEEE Signal Processing Society Young Author Best Paper Award and the Best Paper Award at the IEEE International Conference on Acoustics, Speech, and Signal Processing (ICASSP). She has served as an Elected Member of the SPTM, SAM and MLSP Technical Committees of the IEEE Signal Processing Society. She currently serves as an Associate Editor of IEEE Trans. on Signal Processing.
\end{IEEEbiography}

\end{document}